\documentclass{article}

\usepackage[latin1]{inputenc}
\usepackage{amsmath}%
\usepackage{amsfonts}%
\usepackage{amssymb}%
\usepackage{graphicx}
\usepackage{anysize}
\usepackage{color}
\marginsize{25mm}{25mm}{10mm}{30mm}
\usepackage[sort,round]{natbib}

\newtheorem{theorem}{Theorem}

\newtheorem{lemma}[theorem]{Lemma}

\newtheorem{proposition}[theorem]{Proposition}
\newtheorem{remark}[theorem]{Remark}

\newenvironment{proof}[1][Proof]{\textbf{#1.} }{\ \rule{0.5em}{0.5em}}

\begin{document}

\title{The $\alpha$-Hypergeometric Stochastic Volatility Model\footnote{Comments from participants at the $24^{th}$ New Zealand Econometric Study Group Meeting (2014) held at University of Waikato, Hamilton, New Zealand and at the $10^{th}$ Annual Conference of the Asia-Pacific Association of Derivatives (2014) Busan, Korea, are gratefully acknowledged. The usual caveat applies.}}
\author{Jos\'e Da Fonseca 
\thanks{Auckland University of Technology, Business School, Department of Finance, Private Bag 92006, 1142 Auckland, New Zealand. Phone: ++64 9 9219999 extn 5063. Email: jose.dafonseca@aut.ac.nz} 
\and
Claude Martini \thanks{Zeliade Systems, 56, Rue Jean-Jacques Rousseau, 75001 Paris, France. Phone: ++33 1 40 26 17 81. Email: cmartini@zeliade.com}
}
\date{\today}
\maketitle

\begin{abstract}
The aim of this work is to introduce a new stochastic volatility model for equity derivatives. To overcome some of the well-known problems of the Heston model, and more generally of the affine models, we define a new specification for the dynamics of the stock and its volatility. Within this framework we develop all the key elements to perform the pricing of vanilla European options as well as of volatility derivatives. We clarify the conditions under which the stock price is a martingale and illustrate how the model can be implemented.
\end{abstract}

\vspace{12cm}
\maketitle

\newpage
%\input{Introduction}
%%%%%%%%%%%%%%%%%%%%%%%%%
\section{Introduction}

Positivity is an essential property in financial modelling. In the field of equity derivatives since the seminal work of \cite{BlackScholes1973} several extensions were proposed to handle the stochastic behaviour of the volatility. Among all the proposed models the \cite{heston} model is, certainly, the most analysed model essentially because of its analytical tractability. However, when calibrated on option prices one usually obtains that the Feller condition, ensuring that the process does not reach zero in finite time, is not satisfied (see \cite{DaFonseca2011} for calibration results on several indexes as well as extensions of the Heston model). The mean reverting parameter is problematic to estimate and uses to be small. Notice that this fact seems to be widely known among practitioners and is sometimes mentioned in academic works, see \cite{HenryLabordere2009} page 183. It seems to us that this problem is mainly related to the fact that option prices contain integrated volatility. One consequence is that the volatility remains ``too close'' to zero and contrasts with its empirical distribution which is closer to a lognormal one. This partially motivates the model proposed in \cite{Gatheral2008} which specifies  a (double) lognormal dynamic for the volatility. In that case, a major drawback is that no closed-form solution is available for vanilla options turning the calibration a tedious exercise.\\

From an historical point of view it is well known that the first stochastic volatility model was proposed in \cite{Hull1987} and specifies for the volatility dynamic a geometric Brownian motion which is therefore non stationary. A closed-form solution (for vanilla options) for this model was proposed much later in \cite{leb1996} but it was also shown in \cite{jourdain2004} that the spot loses its martingale property for some parameter values. From a modelling point of view the model of \cite{Chesney1989} is certainly the most natural one as it specifies for the volatility the exponential of a stationary Ornstein-Uhlenbeck process. By construction the volatility is positive. Unfortunately, no closed-form solution for the characteristic function of the stock is available for this model. In \cite{Stein1991} a closed-form solution is obtained for a stochastic volatility model whose volatility follows an Ornstein-Uhlenbeck process, hence  Gaussian, which is problematic regarding the aspect of positivity.\\

Our purpose is to develop a stochastic volatility model which is tractable, that is to say, for which most of the key ingredients to perform derivative pricing can be computed efficiently and has positive distribution for the volatility.\\

The structure of the paper is as follows. In a first section, we introduce the volatility dynamic and study both the volatility and spot properties. For the volatility, we perform a transformation of the process that allows us to heavily use the results of \cite{don2001} whereas for the stock we closely follow \cite{jourdain2004}. In a second section, specifying further the volatility dynamic we analyse the volatility and compute the Mellin transform of the stock using an approach based on the resolvent. For this part, we were inspired by \cite{pin2010} and \cite{Pintoux2011} and heavily used the surveys of \cite{my1} and \cite{my2}. We postpone the discussion of related works to the third section and the last section concludes the paper.

\section{The Model and its Properties}
%%%%%%%%%%%%%%%%%%%%%%%%%%%%%%%%%%%%%%
%%%%%%%%%%%%%%%%%%%%%%%%%%%%%%%%%%%%%%
In the $\alpha$-Hypergeometric model the dynamic is given by

\begin{eqnarray}
df_t&=&f_t e^{v_t }dw_{1,t},\\
dv_t&=&(a - b e^{\alpha v_t}) dt + \sigma dw_{2,t} \label{row2}
\end{eqnarray}

with  $\alpha>0$ and $(w_{1,t},w_{2,t})_{t\geq 0}$ a Brownian motion with $dw_{1,t}.dw_{2,t}=\rho dt$ under the risk neutral probability measure $P$. We denote by $\mathbb{E}$ the expectation under this probability (also, we may denote $\mathbb{E}^P$ whenever needed to avoid confusion). So $v_t$ is the instantaneous log volatility and the instantaneous variance is given by $V_t=e^{2v_t}$. We assume $b >0$ and $\sigma>0$, yet there is no constraint on the sign of $a$.\\

This new equity model dynamic has been designed to make up for the numerous flaws observed when implementing the Heston model (or any other affine model).
%%%%%%%%%%%%%%%%%%%%%%%%%
%%%%%%%%%%\input{Variance}

\subsection{Study of the Variance Process}
%%%%%%%%%%%%%%%%%%%%%%%%%%%%%%%%%%%%%%%
%%%%%%%%%%%%%%%%%%%%%%%%%%%%%%%%%%%%%%%

\subsubsection{The variance as a functional of $w_{2}$ }
%%%%%%%%%%%%%%%%%%%%%%%%%%%%%%%%%%%%%%%%%%%%%%%%%%%%
Let $(v_t)_{t\geq 0}$ denote any solution of the stochastic differential equations Eq.\eqref{row2} (SDE in the sequel). Let us observe that $v_t-v_0 +b \int_0^t \exp{ \alpha v_s} ds =a t +\sigma w_{2,t}$. Introducing the integral $I(t) = \int_0^t \exp{ \alpha v_s} ds$, we note that $\frac{d I(t)}{dt} = \exp{ \alpha v_t}$ so that
$$\ln{\frac{d I(t)}{dt}}+\alpha b  I(t) = \alpha (v_0 + a t +\sigma w_{2,t})$$
or yet
$$\frac{d I(t)}{dt} \exp{\alpha b I(t)} = \exp{\alpha (v_0 + a t +\sigma w_{2,t})},$$
which gives in turn by integrating
$$\exp{\alpha b I(t)} = 1+\alpha b \int_0^t \exp{\alpha (v_0 + a s +\sigma w_{2,s})} ds.$$

We get eventually

$$I(t)=\frac{\ln{ (1+\alpha b \int_0^t \exp{\alpha (v_0 + a s +\sigma w_{2,s})} ds)}}{\alpha b}$$

and by differentiating

$$V_t = \exp{2 v_t} = \left(\frac{d I(t)}{dt}\right)^\frac{2}{\alpha} = \frac{V_0 \exp{2a t + 2\sigma w_{2,t}}}{(1+\alpha b V_0^{\frac{\alpha}{2}} \int_0^t \exp{\alpha (a s +\sigma w_{2,s}) ds})^\frac{2}{\alpha}}.$$

Conversely, let $v$ be defined by the preceding equation (i.e., $v_t = \frac{1}{2}\ln{V_0}+ a t + \sigma w_{2,t}-\frac{\ln{(1+\alpha b V_0^{\frac{\alpha}{2}} \int_0^t \exp{\alpha (a s +\sigma w_{2,s})} ds)}}{\alpha}$).
Then $v_0 = \frac{1}{2}\ln{V_0}$ and
$$\int_0^t \exp{\alpha v_s} ds = \int_0^t \frac{V_0^{\frac{\alpha}{2}} \exp{\alpha (a s + \sigma w_{2,s})}}{1+\alpha b V_0^{\frac{\alpha}{2}} \int_0^s \exp{\alpha (a u +\sigma w_{2,u})} du} ds = \frac{\ln{ (1+\alpha b \int_0^t V_0^{\frac{\alpha}{2}} \exp{\alpha (a s +\sigma w_{2,s})} ds)}}{\alpha b}$$

so that $v_t = v_0 + a t + \sigma w_{2,_t}-b \int_0^t \exp{\alpha v_s} ds$, which is the integrated form of the above SDE.\\  

We have therefore proven that there is existence and pathwise uniqueness for the SDE defining the variance behaviour. Moreover, we have an explicit solution to this SDE in terms of the driving Brownian motion $w_2$. Lastly, observe also that in the limiting case $\alpha=0$, one directly gets $v_t^0 = (a-b) t+ w_{2,t}$. It is easily checked that $\lim_{\alpha \to 0} v_t^\alpha = v_t^0$ pathwise, so that there is no loss of continuity when $\alpha$ converges to zero.

\subsubsection{Basic properties}
%%%%%%%%%%%%%%%%%%%%%%%%%%%%

\paragraph{Dependency on $\alpha$:}
From the driving SDE it is easily seen by scaling that
$$\alpha v_{v_0, \alpha, a, b, \sigma}= v_{\alpha v_0, 1, \alpha a, \alpha b, \alpha \sigma}, V^\alpha_{V_0, \alpha, a, b, \sigma}= V_{V^\alpha_0, 1, \alpha a, \alpha b, \alpha \sigma}  $$
this can be checked also directly on the preceding formulas.

\paragraph{What happens for negative $b$:}
It follows that the SDE has a well defined solution when $b$ {\it and} $\alpha$ are negative. If $b<0$ and $\alpha >0$, it follows form the expression of $I(t)$ that the solution is well defined up to the stopping time 
$$T^{*}=\inf \left\lbrace t \vert \int_0^t \exp{\alpha (v_0 + a s +\sigma w_{2,s})} ds > -\frac{1}{\alpha b}\right\rbrace$$

\paragraph{Noiseless limit:} \label{sec:noiseless}
The above computations are valid when 
$\sigma=0$. In this case the formula simplifies to

$$I(t)=\frac{\ln{ (1+\frac{b}{a}e^{\alpha v_0} (e^{\alpha a t}-1))}}{\alpha b}$$

and by differentiating

$$V_t = \frac{V_0 e^{2a t}}{(1+\frac{b}{a} V_0^{\frac{\alpha}{2}} (e^{\alpha a t}-1))^\frac{2}{\alpha}}.$$

It follows in particular that $\frac{I(t)}{t} \to \frac{a}{b}$ when $t \to \infty$.

\subsubsection{Connection with the Wong-Shyryaev process}
%%%%%%%%%%%%%%%%%%%%%%%%%%%%%%%%%%%%%%%%%%%%%%%%%%%%%%

The driving process of the variance is

$$dv_t = (a-b e^{\alpha v_t}) dt + \sigma dw_{2,t}.$$

Consider now $Z_t = e^{-\alpha v_t}=V_t^{-\frac{\alpha}{2}}$ then we have 
$$d(e^{-\alpha v_t}) = -\alpha e^{-\alpha v_t} [(a-b e^{\alpha v_t}) dt + \sigma dw_{2,t}] + \frac{\alpha^2 \sigma^2}{2} e^{-\alpha v_t} dt$$
so that
$$dZ_t = [\alpha b +(\frac{\alpha^2 \sigma^2}{2} -\alpha a) Z_t] dt + \alpha \sigma Z_t dC_t$$
with $Z_0=e^{-\alpha v_0} $ where $(C_t)_{t\geq 0}$ is the Brownian motion $(-w_{2,t})_{t\geq 0}$.
This process (modulo a convenient rescaling) has been studied in \cite{don2001} and in \cite{peskir} where it is called the Shyryaev process. It is also sometimes called the Wong process. We will mostly make use of \cite{don2001}. 
To alleviate the notation, let us write it as
$$dZ_t = (m+n Z_t) dt + p Z_t dC_t$$

with

$$m=\alpha b,\quad n=(\frac{\alpha^2 \sigma^2}{2} -\alpha a) \textrm{ and } p= \alpha \sigma. $$

This SDE is affine in $Z_t$ and thus easy to solve: introduce $(X_t)_{t \geq 0}$ the solution of the homogeneous SDE $dX_t = n X_t dt + p X_t dC_t$, so that
$X_t = X_0 e^{(n-\frac{p^2}{2}) t + p C_t}$, and look for solutions of the form $U_t X_t$ where $(U_t)_{t \geq 0}$ is a process of finite variation. Then $d(U_t X_t)=U_t (n X_t dt + p X_t dC_t) + X_t dU_t$, so that $U_tX_t$ is a solution as soon as $dU_t = m X^{-1}_t dt$. As a result we have
$$Z_t = X_t (Z_0 + m  \int_0^t X^{-1}_s ds)$$
with $X_s = e^{(n-\frac{p^2}{2}) s + p C_s}.$ 

Now let us look for a scaling factor $c$ such that $p C_{c u} = \sqrt{2} D_u$ for some Brownian motion $(D_t)_{t\geq 0}$. Obviously, we have 
$c=\frac{2}{p^2}$ and
$$X_u = e^{(n-\frac{p^2}{2}) c \frac{u}{c}   + \sqrt{2} D_\frac{u}{c}}.$$
So, within the notations of \cite{don2001}  we deduce that
$$\nu = (n-\frac{p^2}{2}) c=-\alpha a c = -a \frac{2}{\alpha \sigma^2},$$

$$Z_t = X_{c \frac{t}{c}} (Z_0 + m c  \int_0^{\frac{t}{c}} X^{-1}_{c u} du) = m c Y^{(\nu)}_\frac{t}{c}(\frac{Z_0}{m c})$$
where
$$m c = \alpha b \frac{2}{p^2} = \frac{2 b}{\alpha \sigma^2}$$ 
and $Y_t^{\nu}(x)$ satisfies the SDE
\begin{align}
Y_t=x+\sqrt{2}\int_0^t Y_udB_u + \int_0^t (1 + (1+\nu)Y_u)du
\end{align}

where $(B_t)_{t\geq 0}$ a Brownian motion. Combining these results we have the following proposition relating the volatility process of the model to the process $Y_t^{\nu}(x)$.
\begin{proposition}The instantaneous volatility $V_t$ can be expressed as a function of the process $Y_t^{(\nu)}$ through the relation
$$e^{-\alpha v_t}=V_t^{-\frac{\alpha}{2}} = q Y^{(\nu)}_\frac{t}{c}(\frac{V_0^{-\frac{\alpha}{2} }}{q} )$$
with
$$\nu = - \frac{2a}{\alpha \sigma^2},\quad  q = \frac{2 b}{\alpha \sigma^2} \textrm{ and }  c = \frac{2}{\alpha \sigma^2}.$$

\end{proposition}

%%%%%%%%%%%%%%%%%%%%%%%%%%%%%%%%%%%%%%%%%%%%%%%%%%%%%%%%%%%
%%%%%%%%%%%%%%%%%%%%%%%%%%%%%%%%%%%%%%%%%%%%%%%%%%%%%%%%%%%
\paragraph{The Green function of $V$:}

The Green function $u_{\lambda}(x, y) $ associated to $Y_t^{(\nu)}(x)$ has been computed in Theorem 3.1 \cite{don2001}, so we have the Green function of $v_t$ or $V_t$:

$$u_{\lambda}(x, y) = \frac{\Gamma(\frac{\mu+\nu}{2})}{\Gamma(1+\mu)} (\frac{1}{y})^{1-\nu} e^{-\frac{1}{y}} [1(y\leq x)
\phi_1(x) \phi_2(y)+ 1(x<y) \phi_2(x) \phi_1(y)]  $$
where 
$$\mu=\sqrt{\nu^2+4 \lambda}$$
and
$$\phi_1(x)= \left(\frac{1}{x}\right)^{\frac{\mu+\nu}{2}} \Phi\left(\frac{\mu+\nu}{2}, 1+\mu;\frac{1}{x}\right),  $$

$$\phi_2(x)= \left(\frac{1}{x}\right)^{\frac{\mu+\nu}{2}} \Psi\left(\frac{\mu+\nu}{2}, 1+\mu;\frac{1}{x}\right).  $$

The function $\Phi$ is the confluent hypergeometric function of the first kind, which has the integral representation:
\begin{align*}
\Phi(\alpha, \gamma;z) = \frac{\Gamma(\gamma)}{\Gamma(\alpha) \Gamma(\gamma-\alpha)} \int_0^1 e^{z u} u^{\alpha-1} (1-u)^{\gamma-\alpha-1} du, 
\end{align*}

and $\Psi$ is the confluent hypergeometric function of the second kind. We will use the following integral representation, called the Barnes integral representation, and given by (\cite{DLMF} 13.4.18)
\begin{align}
\Psi(a-1, b;z) = \frac{z^{1-b} e^z}{2i\pi} \int_{-i\infty}^{i\infty} \frac{\Gamma(b-1+t) \Gamma(t)}{\Gamma(a-1+t)} z^{-t} dt \label{BarnesPsi}
\end{align}

where the contour of integration passes on the right of the poles of the integrand. The function $\Phi$ is also denoted $M$ or ${}_1F_1$ and also called the Kummer function while $\Phi$ is also written $U$ and is the Tricomi function.

\paragraph{The moments of $V^{-\frac{\alpha}{2}}$:}

It is also easy to compute the $l^{th}$ moment $M^{(l)}$ of $Z$ as we have
$$dZ^l_t = l Z^{l-1}_t ((m+n Z_t) dt + p Z_t dC_t) + l (l-1)Z^{l-2}_t p^2 Z^2_t dt$$
whence
$$dM^{(l)}_t = (m l M^{(l-1)}_t + (n l+l(l-1)p^2) M^{(l)}_t)dt.$$

%In particular by the same computation as above
%$$M^{(1)}_t = \mathbb{E}\left[V^{-\frac{\alpha}{2}}_t\right] = e^{-\alpha a t} (V^{-\frac{\alpha}{2}}_0 + \alpha b \int_0^t e^{\alpha a s} ds )$$

\subsubsection{The pricing of Variance Swaps}
%%%%%%%%%%%%%%%%%%%%

We are interested in the following quantity
$$t \textsc{vs}(t) = \int_0^t \mathbb{E}[ V_s] ds = \int_0^t \mathbb{E}\left[ Z^{-\frac{2}{\alpha}}_s\right] ds = q^{-\frac{2}{\alpha}} \int_0^t \mathbb{E}[\frac{1}{Y^{(\nu)}_\frac{s}{c}(\frac{V_0^{-\frac{\alpha}{2}}}{q})^\frac{2}{\alpha}}] ds = c q^{-\frac{2}{\alpha}} \int_0^\frac{t}{c} \mathbb{E}[\frac{1}{Y^{(\nu)}_r(\frac{V_0^{-\frac{\alpha}{2}}}{q})^\frac{2}{\alpha}}] dr.$$
It is involved in the pricing of variance swap that is an important financial volatility product.  
\paragraph{The Laplace transform of $\textsc{vs}(t)$ via the Green function:}
%%%%%%%%%%%%%%%%%%%%%%%%%%%%%%%%%
By using  the standard algebraic operations on Laplace transforms we know that the Laplace transform of the integral is the Laplace transform of the integrand divided by $\lambda$. So the first step is to compute

$$I(x) = \int_0^\infty e^{-\lambda r} \mathbb{E}[\frac{1}{Y^{(\nu)}_r(x)^\frac{2}{\alpha}}] dr = \int_0^\infty y^{-\frac{2}{\alpha}} u_{\lambda}(x, y) dy. $$

We have $\frac{\Gamma(1+\mu)}{\Gamma(\frac{\mu+\nu}{2})} I = \phi_1(x) I_2 + \phi_2(x) I_1 $ with
$$I_1=\int_x^\infty  \left(\frac{1}{y}\right)^{\frac{2}{\alpha}+1-\nu} e^{-\frac{1}{y}} \phi_1(y) dy = \int_x^\infty  \left(\frac{1}{y}\right)^{\frac{2}{\alpha}+1+\frac{\mu-\nu}{2}} e^{-\frac{1}{y}} \Phi\left(\frac{\mu+\nu}{2}, 1+\mu;\frac{1}{y}\right) dy,$$

$$I_2=\int_0^x  \left(\frac{1}{y}\right)^{\frac{2}{\alpha}+1-\nu} e^{-\frac{1}{y}} \phi_2(y) dy= \int_0^x  \left(\frac{1}{y}\right)^{\frac{2}{\alpha}+1+\frac{\mu-\nu}{2}} e^{-\frac{1}{y}} \Psi\left(\frac{\mu+\nu}{2}, 1+\mu;\frac{1}{y}\right) dy.$$

Then
$$\textsc{vs}(t)= \frac{c}{q t} \mathcal{L}^{-1}\left(\frac{1}{\lambda} I(x=\frac{1}{q}V_0^{-\frac{\alpha}{2}})\right)(\frac{t}{c})$$
where $\mathcal{L}^{-1}$ denotes the inverse Laplace transform.

%%%%%%%%%%%%%%%%%%%%%%%%%%%%%%%%%%%%%%%%%%%%%%%%%
\subsubsection{Computation of $I_1$}
By the change of variable $z=\frac{1}{y}$, $I_1=\int_0^{{\frac{1}{x}}}  z^{\frac{\mu-\nu}{2}+\frac{2}{\alpha}-1} e^{-z} \Phi(\frac{\mu+\nu}{2}, 1+\mu;z) dz$. Let us introduce $a=1+\frac{\mu+\nu}{2}, b=1+\mu$, so that $b-a=\frac{\mu-\nu}{2}$. By Kummer's tranformation
$e^{-z} \Phi(a-1, b;z)=\Phi(b-a+1, b;-z)$ and 
$$I_1=\int_0^{\frac{1}{x}} z^{b-a+\frac{2}{\alpha}-1}  \Phi(b-a+1, b;-z) dz.$$
Therefore, by Fubini's theorem as in \cite{fredholm}, $I_1=\sum_{n=0}^\infty \frac{(b-a+1)_n}{(b)_n n!} (-1)^n \int_0^{\frac{1}{x}} z^{b-a+\frac{2}{\alpha}-1+n} dz,$ which leads to
$$I_1= x^{-b+a-\frac{2}{\alpha}} \sum_{n=0}^\infty \frac{(b-a+1)_n}{(b-a+\frac{2}{\alpha}+n) (b)_n n!} (-1)^n x^{-n} $$
Note that $(b-a+\frac{2}{\alpha}+n) = \frac{(b-a+\frac{2}{\alpha}+1)_n (b-a+\frac{2}{\alpha})}{(b-a+\frac{2}{\alpha})_n}$ so that eventually
\begin{align*}
I_1&= \frac{ x^{-b+a-\frac{2}{\alpha}}}{(b-a+\frac{2}{\alpha})} \sum_{n=0}^\infty \frac{(b-a+1)_n (b-a+\frac{2}{\alpha})_n}{(b-a+\frac{2}{\alpha}+1)_n (b)_n n!} (-1)^n x^{-n} \\
&= \frac{ x^{-b+a-\frac{2}{\alpha}}}{(b-a+\frac{2}{\alpha})} H\left(\left[b-a+1, b-a+\frac{2}{\alpha}\right], \left[b, b-a+1+\frac{2}{\alpha}\right], -\frac{1}{x}\right)
\end{align*}
 
where $H$ is the {\it generalized hypergeometric} function.

%%%%%%%%%%%%%%%%%%%%%%%%%%%%%%%%%%%%%%%%%%%%%%%%
\subsubsection{Computation of $I_2$}
%%%%%%%%%%%%%%%%%%%%%%%%%%%%%%%%%%%%%%%%%%%%%%%%%%%%%%%%%
\paragraph{$I_2$ as a complex integral.}
We have in the same way $I_2=\int_{\frac{1}{x}}^\infty z^{b-a+\frac{2}{\alpha}-1} e^{-z} \Psi(a-1, b;z)dz$. Thanks to the Barnes integral representation Eq.\eqref{BarnesPsi}   we have 
$$e^{-z} z^{b-a+\frac{2}{\alpha}-1} \Psi(a-1, b;z) = \frac{1}{2i\pi} \int_{-i\infty}^{i\infty} \frac{\Gamma(b-1+t) \Gamma(t)}{\Gamma(a-1+t)} z^{\frac{2}{\alpha}-(a+t)}  dt.  $$
Moreover we know that the integral converges locally uniformly in $z$, so that we can apply Fubini's theorem and permute the integrals.
Observe now that $a=1+\frac{\mu+\nu}{2}=1+\frac{\sqrt{\nu^2+4\lambda}-\left|\nu\right|}{2}$, so that $1-a<0$.
If $\lambda$ is large enough so that $1+\frac{2}{\alpha}-a<0$, then the inner integral is finite and
$$I_2 = - \frac{1}{2i\pi} \int_{-i\infty}^{i\infty} \frac{\Gamma(b-1+t) \Gamma(t)}{\Gamma(a-1+t) (1+\frac{2}{\alpha}-(a+t))} {(\frac{1}{x})}^{1+\frac{2}{\alpha}-(a+t)} dt.  $$
This formula is valid as soon as $a>1+\frac{2}{\alpha}$, which amounts after a simple computation to $$\lambda>\lambda^*=\frac{4}{\alpha^2}+\frac{2\left|\nu\right|}{\alpha}.$$

Writing $(1+\frac{2}{\alpha}-(a+t)) = \frac{\Gamma(2+\frac{2}{\alpha}-(a+t))}{\Gamma(1+\frac{2}{\alpha}-(a+t))}=\frac{\Gamma(1-(a-\frac{2}{\alpha}-1)-t)}{\Gamma(1-(a-\frac{2}{\alpha})-t)}$ and recalling the definition of Meijer $G$ function, $I_2$ looks like
$$-{(\frac{1}{x})}^{1+\frac{2}{\alpha}-a} G\left([[a-\frac{2}{\alpha}], [a-1]], [[0, b-1], [a-\frac{2}{\alpha}-1]], \frac{1}{x}\right).  $$
Nevertheless, the paths of integration are not the same for the two formulas, since the defining path in the Meijer $G$ function is not on the right of all the poles of the integrand.

%%%%%%%%%%%%%%%%%%%%%%%%%%%%%%%%%%%%%%%%%%%%%%%%%%%%%%%%%%%%%%%%%%%%
\paragraph{An explicit hypergeometric series for $I_2$:}
At this stage the natural step is to apply the theorem of residues to get a series from the above complex integrals. The poles of the integrand are located:
\begin{itemize}
\item at $t=1+\frac{2}{\alpha}-a (<0)$, with residue $-\frac{\Gamma(b-1+t) \Gamma(t)}{\Gamma(a-1+t)}$. 
\item at $t=-n, n \in \mathbb{N}$, with residue $\frac{\Gamma(b-1+t)}{\Gamma(a-1+t)(1+\frac{2}{\alpha}-(a+t)) n!} (-1)^n y^{1+\frac{2}{\alpha}-(a+t)}$.
\item at $t=-n+1-b, n \in \mathbb{N}$, with residue $\frac{\Gamma(t)}{\Gamma(a-1+t)(1+\frac{2}{\alpha}-(a+t)) n!} (-1)^n y^{1+\frac{2}{\alpha}-(a+t)}$.
\end{itemize}
so that by Cauchy's residue theorem we get
\begin{align}
I_2 &=\frac{\Gamma(b-a +\frac{2}{\alpha}) \Gamma(1+\frac{2}{\alpha}-a)}{\Gamma(\frac{2}{\alpha})}\nonumber\\
&+\sum_{n=0}^\infty \frac{(-1)^{n+1} y^{1+\frac{2}{\alpha}+n-a}}{n!} \left(\frac{\Gamma(b-1-n)}{\Gamma(a-1-n)(1+\frac{2}{\alpha}+n-a)} + \frac{y^{b-1} \Gamma(1-b-n)}{\Gamma(a-b-n)(\frac{2}{\alpha}+b-a+n)}\right) 
\end{align}
and $I_2$ can be computed easily by making explicit the recurrences between successive terms of the two series. Calls to the $\Gamma$ function are only required for the constant and index zero terms.

%%%%%%%%%%%%%%%%%%%%%%%%%%%%%%%%%%%%%%%%%%%%%%%%%%%%%%%%%%%%%%%%%%%%
\subsubsection{Final formula for I}
Since $I(x) = \frac{\Gamma(\frac{\mu+\nu}{2})}{\Gamma(1+\mu)}(\phi_1(x) I_2(x) + \phi_2(x) I_1(x)) $ we get the final formula for $I$ given by the following proposition.

\begin{proposition}
For any $\lambda>\lambda^*$ where $\lambda^*=\frac{4}{\alpha^2}+\frac{2\left|\nu\right|}{\alpha}$,
$$I=\frac{\Gamma(a-1)}{\Gamma(b)}(y^{a-1} I_2  \Phi(a-1, b, y)  +y^{b+\frac{2}{\alpha}-1} \frac{\Psi(a-1, b, y)}{(b-a+\frac{2}{\alpha})} h)$$
where $a=1+\frac{\mu+\nu}{2}, b=1+\mu$,  $y=\frac{1}{x}$ and
$$h=H([b-a+1, b-a+\frac{2}{\alpha}], [b, b-a+1+\frac{2}{\alpha}], -y) $$
\begin{align}
I_2 &= \frac{\Gamma(b-a +\frac{2}{\alpha}) \Gamma(1+\frac{2}{\alpha}-a)}{\Gamma(\frac{2}{\alpha})}\nonumber\\
&+\sum_{n=0}^\infty \frac{(-1)^{n+1} y^{1+\frac{2}{\alpha}+n-a}}{n!} \left(\frac{\Gamma(b-1-n)}{\Gamma(a-1-n)(1+\frac{2}{\alpha}+n-a)} + \frac{y^{b-1} \Gamma(1-b-n)}{\Gamma(a-b-n)(\frac{2}{\alpha}+b-a+n)}\right).
\end{align} 
\end{proposition}

\subsubsection{Short term behaviour}

We now analyse the short term behaviour of the instantaneous volatility. We start from the formula
$$V_t = \frac{V_0 \exp{2a t + 2\sigma w_{2,t}}}{(1+\alpha b V_0^\frac{\alpha}{2} \int_0^t \exp{\alpha (a s +\sigma w_{2,s}) ds})^\frac{2}{\alpha}}.$$
By introducing  the exponential martingale $e^{2\sigma w_{2,t} - \frac{(2\sigma)^2}{2}t}$ we get by Girsanov's theorem
$$\mathbb{E}[V_t] = V_0 e^{(2a+ \frac{(2\sigma)^2}{2}) t}  \mathbb{E}^Q\left[\left( 1+\alpha b V_0^\frac{\alpha}{2} \int_0^t \exp{\alpha ((a+ \sigma^2) s +\sigma \tilde{w}_{2,s}) ds}\right)^{\frac{-2}{\alpha}} \right]$$ with $\tilde w_{2,t}=w_{2,t}-2\sigma t$ a Brownian motion under $Q$.
For a given $t$, the set of paths such that the time integral is larger than an arbitrary small level becomes exponentially small in probability so that
$$\mathbb{E}[V_t] \sim V_0 e^{(2a+ \frac{(2\sigma)^2}{2}) t} \mathbb{E}^Q\left[1-2b V_0 ^\frac{\alpha}{2} \int_0^t \exp{\alpha ((a+ \sigma^2) s +\sigma \tilde w_{2,s}) ds}\right].$$

Now $\mathbb{E}^Q[\int_0^t \exp{\alpha ((a+ \sigma^2) s +\sigma \tilde{w}_{2,s})} ds] = \frac{e^{(\alpha (a + \sigma^2) +\frac{\alpha^2 \sigma^2}{2})t} -1}{\alpha (a+\sigma^2)+\frac{\alpha^2 \sigma^2}{2}}$.
Therefore,  in the following proposition the second statement results from the first one by integration.
\begin{proposition}
As $t \to 0$
\begin{itemize}
\item $\mathbb{E}[V_t] \sim V_0 (1+(2a+\frac{(2\sigma)^2}{2}-2bV_0^{\frac{\alpha}{2}}) t)$
\item $\textsc{vs}(t) \sim V_0 (1+(2a+\frac{(2\sigma)^2}{2}-2bV_0^{\frac{\alpha}{2}})\frac{t}{2} )$
\end{itemize}

\end{proposition}

\paragraph{Short term behaviour when $\alpha=2$:}
There is an easy majorization in case $\alpha=2$, which also provides an excellent approximation for short term maturities: by using the concavity of the logarithm and Jensen's inequality
$$t \textsc{vs}(t) = \mathbb{E}^P[\frac{1}{2b}\ln( 1+2b V_0 A_t)] < \frac{1}{2b} \ln( 1+2b V_0  \mathbb{E}^P[A_t])$$

where $\mathbb{E}^P[A_t] = \frac{e^{(2a+2\sigma^2)t} -1}{(2a+2\sigma^2)}$. This will yield an excellent short term approximation because $\ln{1+x} \sim x$ near $0$ and  $A_t$ is small in probability for small $t$.
\begin{proposition} ($\alpha=2$)
Let $f(t) =\frac{1}{2bt} \ln{( 1+2b V_0 \frac{e^{(2a+2\sigma^2)t} -1}{(2a+2\sigma^2)}})	$.
Then $\textsc{vs}(t)<f(t)$ for every $t>0$. Moreover as $t \sim 0$,
$$\textsc{vs}(t) \sim f(t) \sim V_0 (1+(a+\sigma^2-bV_0)t).$$
\end{proposition}
The last approximation is useful for practical purposes.
%%%%%%%%%%%%%%%%%%%%%%%%%%%%%%%%%%%%%%%%%%%%%%%%%%%%%%%%%%%
\subsubsection{Long term behaviour when $a>0$}
%%%%%%%%%%%%%%%%%
We start also from the formula
$$V_t = \frac{V_0 \exp{2a t + 2\sigma w_{2,t}}}{(1+\alpha b V_0^{\frac{\alpha}{2}} \int_0^t \exp{\alpha (a s +\sigma w_{2,s}) ds})^\frac{2}{\alpha}}.$$
Since $a>0$, the behaviour of the average $\int_0^t e^{\alpha (a s +\sigma w_{2,s})} ds$ will go very fast to infinity as $t \to \infty$. It is clear in particular that $\int_0^t e^{\alpha (a s +\sigma w_{2,s})} ds$ will become much larger than 1 so that 
$$V_t \sim \frac{V_0 e^{2a t + 2\sigma w_{2,t}}}{(\alpha b V_0^\frac{\alpha}{2} \int_0^t \exp{\alpha (a s +\sigma w_{2,s}) ds})^\frac{2}{\alpha}}.$$
Now this simplifies to $\frac{1}{(\alpha b)^\frac{2}{\alpha}} \frac{1}{(\int_0^t \exp{\alpha (a (s-t) +\sigma (w_{2,s}-w_{2,t})} ds)^\frac{2}{\alpha}}$, and 
the whole point is to observe that by time-reversal we will get an average with a {\it negative} drift, whose behaviour at infinity converges to the inverse of a Gamma law:  by scaling $\int_0^t e^{\alpha (a (s-t) +\sigma (w_{2,s}-w_{2,t})} ds \overset{d}{=} \frac{4}{\alpha^2 \sigma^2} \int_0^{t \alpha^2 \frac{\sigma^2}{4}} e^{2(-\frac{2a}{\alpha \sigma^2}u +B_u)}du$ for some Brownian motion $B$. With the notations of \cite{dufresne1998} with a drift $\mu=\frac{2a}{\alpha \sigma^2}$ we have therefore
$$V_t \sim (\alpha b)^{-\frac{2}{\alpha}} ( \frac{2}{\alpha^2 \sigma^2})^{-\frac{2}{\alpha}} (2 A^{(-\mu)}_{t \alpha^2 \frac{\sigma^2}{4}})^{-\frac{2}{\alpha}},$$
which entails
$$V_t \to (\alpha b)^{-\frac{2}{\alpha}} ( \frac{2}{\alpha^2 \sigma^2})^{-\frac{2}{\alpha}}(\text{Gamma}(\mu, 1))^\frac{2}{\alpha}.$$
Observing that the expectations will converge too thanks to the monotone convergence theorem, we obtain the following proposition.
\begin{proposition}
Assume $a>0$. As $t \to \infty$, $$\mathbb{E}[V_t], \textsc{vs}(t) \to  ( \frac{2b}{\alpha \sigma^2})^{-\frac{2}{\alpha}}\mathbb{E}[(\text{Gamma}(\mu, 1))^\frac{2}{\alpha}]$$
where $\mu=\frac{2a}{\alpha \sigma^2}$.
In particular, \begin{itemize}
\item For $\alpha=2$, $\mathbb{E}[V_t], \textsc{vs}(t) \to  \frac{a}{b}$.
\item For $\alpha=1$, $\mathbb{E}[V_t], \textsc{vs}(t) \to  (\frac{\sigma^2}{2 b})^2 \frac{2a}{\sigma^2}(1+\frac{2a}{\sigma^2})$.
\end{itemize}
\end{proposition}

Note that this is consistent with the large time behaviour of the noiseless limit obtained in \ref{sec:noiseless} for the case $\alpha=2$. The noiseless limit in the above formula for $\alpha=1$ is $(\frac{a}{b})^2$, which is not that of $I(t)$ which is $\frac{a}{b}$ irrespective of $\alpha$: just note that $I(t)$ is {\it not} the integrated variance when $\alpha \neq 2$, so there is no contradiction or mysterious lack of continuity behaviour.

%%%%%%%%%%%%%%%%%%%%%%%%%%%%%%%%%%%%%
%%%%%%%%%%%%%%%%%%%%%%%%%%%%%%%%%%%%%%
%%%%%%%%%%%%%%%%%%%%%%%%%%%%%%%%%%%%%%%%%%%%%%%%%%%%%%%
%\clearpage
%%\input{Spot}
%%%%%%%%%%%%%%%%%%%%%%%%%%
\subsection{Study of the Spot Process}
%%%%%%%%%%%%%%%%%%%%%%%%%%%%%&&&&&&
%%%%%%%%%%%%%%%%%%%%%%%%%%%%%&&&&&&

\subsubsection{A full-blown martingale}
%%%%%%%%%%%%%%%%%%%%%%%%%%%%%%%%%%%%

Consider now the dynamic of the forward $(f_t)_{t \geq 0}$, it is defined by the stochastic exponential of the local martingale $L_t = \int_0^t e^{v_s}dw_{1,s}$. Then $<L>_t =\int_0^t e^{2v_s} ds$ and  Novikov's criterion tells us that $(f_t)_{t \geq 0}$ is a uniformly integrable martingale if $\mathbb{E}[\exp{\frac{<L>_t}{2}}]<\infty$. 

\paragraph{Case $\alpha=2$:} In this case

$$\exp{\frac{<L>_t}{2}} = \exp{\frac{1}{2} \int_0^t e^{2v_s} ds} = \exp{\frac{I(t)}{2}} = (1+2b \int_0^t \exp{2(v_0 + a s +\sigma w_{2,s})} ds)^{\frac{1}{4b}}.$$

Therefore, assuming $b>0$, $\exp{\frac{<L>_t}{2} }<(1+\frac{b}{\left|a\right|}  V_0 \exp{2\left|a\right| t} \exp{2\sigma w_{2,t}^*})^{\frac{1}{4b}}$ where $w_{2}^*$ denotes the running maximum of the Brownian motion $w_{2}$. 

Now $\exp{2\sigma w_{2,t}^*}\geq1$ and this is less than $\max{(1, \frac{b}{a}  V_0 \exp{2a t})}^{\frac{1}{4b}} \exp{\frac{2\sigma w_{2,t}^*}{4b}}.$
Since $w_{2,t}^*$ has the same law as $|w_{2,t}|$, this is integrable and 
$$\mathbb{E}[\exp{\frac{<L>_t}{2} }]<\infty.$$

\paragraph{Case $\alpha>2$:} In this case we have $<L>_t=\int_0^t e^{2v_s} ds\leq t^{1-\frac{2}{\alpha}} (\int_0^t e^{\alpha v_s} ds)^\frac{2}{\alpha}$ by Holder's inequality. Now
$$\int_0^t e^{\alpha v_s} ds = \frac{\ln{ (1+\alpha b \int_0^t V_0^\frac{\alpha}{2} \exp{\alpha (a s +\sigma w_{2,s})} ds)}}{\alpha b}.$$
To conclude note that since $\alpha>2$, $(\int_0^t e^{\alpha v_s} ds)^\frac{2}{\alpha}\leq \max{(1, \int_0^t e^{\alpha v_s} ds)}^\frac{2}{\alpha}\leq \max{(1, \int_0^t e^{\alpha v_s} ds)}$
and the equality
$$\max{(1, z)} = z+(1-z) 1(z<1)$$ 
tells us that $e^{\max{(1, \int_0^t e^{\alpha v_s} ds)}}\leq e^{\int_0^t e^{\alpha v_s} ds}e^1$ and we can conclude as above.

\paragraph{Case $\alpha<2$:} In this case, the mean reversion force is weaker and we expect that the log volatility may become large, and therefore also the forward $f$ in case of positive correlation.\\

We follow step-by-step the reasoning of \cite{jourdain2004}. First note that $\mathbb{E}[f_t] = f_0 \mathbb{E}[\mathcal{E}(\rho  \int_0^t \exp{v_s} dw_{2,s})]$.
Since $(f_t)_{t \geq 0}$ is a positive local martingale, it is a supermartingale and the map $t \rightarrow \mathbb{E}[f_t]$ is non-increasing. Therefore, $f_t$ is a martingale if and only if it is constantly equal to $f_0$, i.e. $\frac{\mathbb{E}[f_t]}{f_0}=1$. The quantity
$$\mathbb{E}[\mathcal{E}(\rho  \int_0^t \exp{v_s} dw_{2,s})]$$ turns out to be the probability of non explosion of a Markovian SDE associated to the initial one by means of Girsanov's theorem, and Feller's criterion for explosion provides then an explicit necessary and sufficient condition for this probability to be one.\\

Adopting for a while the notations of \cite{jourdain2004}, we denote $(w_{2,t})_{t \geq 0}$ by $(B_t)_{t \geq 0}$. Introduce the probability $Q$ under which $d\tilde{B}_t = dB_t - \frac{a}{\sigma} dt + \frac{b}{\sigma} V_0^\frac{\alpha}{2} \exp{\alpha \sigma B_t} dt$ is a Brownian motion. By Girsanov's theorem,
$\frac{dQ}{dP} = \mathcal{E}(L_T)$ with $L_t = \frac{a}{\sigma}  B_t -  \frac{b}{\sigma} V_0^\frac{\alpha}{2} \int_0^t \exp{\alpha \sigma B_s} dB_s$. By the Yamada-Watanabe theorem, the law of $(v, B)$ under $P$ is the same as the law of $(v_0+ \sigma B, \tilde{B})$ under $Q$, and
$$\mathbb{E}^P[\mathcal{E}(\rho \int_0^t \exp{v_s} dB_s)]=\mathbb{E}^Q[\mathcal{E}(\rho \sqrt{V_0} \int_0^t \exp{\sigma B_s} d\tilde{B}_s)].$$
This is equal to
$$\mathbb{E}^P[\mathcal{E}(\rho \sqrt{V_0} \int_0^t \exp{\sigma B_s} (dB_s  + (\frac{b}{\sigma} V_0^\frac{\alpha}{2} \exp{\alpha \sigma B_s}- \frac{a}{\sigma}) ds))
\mathcal{E}(\frac{a}{\sigma}  B_t -  \frac{b}{\sigma} V_0^\frac{\alpha}{2} \int_0^t \exp{\alpha \sigma B_s} dB_s)]$$
which rewrites as
$$\mathbb{E}^P[\mathcal{E}(\int_0^t (\rho \sqrt{V_0} \exp{\sigma B_s}-\frac{b}{\sigma} V_0^\frac{\alpha}{2} \exp{\alpha \sigma B_s}+\frac{a}{\sigma}) dB_s)]=\mathbb{E}^P[\mathcal{E}(\int_0^t b(B_s) dB_s)]$$ %(la je ne comprends pas pourquoi la partie en dt part)
with 

$$b(z) = \rho \sqrt{V_0} \exp{\sigma z}-\frac{b}{\sigma} V_0^\frac{\alpha}{2} \exp{\alpha \sigma z}+\frac{a}{\sigma}.$$

The next step is to observe that $\mathbb{E}^P[\mathcal{E}(\int_0^t b(B_s) dB_s)] = P(\tau_\infty>t)$ where $\tau_\infty$ is the explosion time of the SDE
$$dZ_s = b(Z_s) ds + dB_s.$$

We can now apply the Feller criterion for explosions, which tells us that $P(\tau_\infty=\infty)=1$ if and only if
$$a(-\infty) = a(\infty) = \infty$$ where
$a(z)=\int_0^z p'(x) \int_0^x \frac{2}{p'(y)} dy dx$ where $p$ is any scale function of the process $(Z_t)_{t \geq 0}$.

Now the function $p(x) = \int_0^x \exp{-2 \int_0^y b(z) dz} dy$ is a scale function, and we are left with explicit computations.

Observe that $\int_0^y b(z) dz = \frac{\rho \sqrt{V_0}}{\sigma} (\exp{\sigma y}-1) -\frac{b V_0^\frac{\alpha}{2}}{\alpha \sigma^2} (\exp{\alpha \sigma y}-1) + \frac{a}{\sigma} y$
so that
$$p'(x) = C \exp{\left(-2 \frac{\rho \sqrt{V_0}}{{\sigma}} \exp{\sigma x}+2 \frac{b V_0^\frac{\alpha}{2}}{\alpha \sigma^2} \exp{\alpha \sigma x} -2 \frac{a}{\sigma} x\right)}$$ for some positive constant $C$.

\paragraph{Behaviour at $-\infty$:} $p'(x) \sim C \exp{-2 \frac{a}{\sigma}} x$  with also $\int_x^0 \frac{2}{p'(y)} dy$ which is positive and increasing as $x \rightarrow -\infty$, so that $a(-\infty) = \infty$ when $a>0$. This argument is still valid when $a=0$. When $a<0$, then
$\int_x^0 \frac{2}{p'(y)} dy \sim C^{-1} \exp{2 \frac{a}{\sigma}} x$ so that the integrand converges to the constant 1 and the integral diverges.

\paragraph{Behaviour at $+\infty$:} There again, $\int_x^0 \frac{2}{p'(y)} dy$ is positive and increasing as $x \rightarrow \infty$. 
The behaviour is driven by the terms in the outer exponential:
\begin{itemize}
\item When $\rho\leq0$, the exponential terms will dominate the linear one, and $a(\infty)=\infty$.
\item When $\rho>0$ and $\alpha > 1$, the second positive exponential will dominate the first negative one, and $a(\infty)=\infty$.
\item When $\alpha =1$, $p'(x)$ writes
$$C \exp{\left(\frac{2\sqrt{V_0}}{\sigma^2} (b-\rho \sigma) \exp{\sigma x}- \frac{2a}{\sigma} x\right)}$$
so that if $b > \rho \sigma$, $a(\infty)=\infty$. A straightforward computation shows that this also holds when
$b = \rho \sigma$.
\end{itemize}

The other cases require a little more work. So let us assume $\alpha\leq 1$ and $\rho>0$. Then $p'(x)$ rewrites
$$\exp{A\exp{\alpha x} - B \exp{x} -D x}$$ for positive $A,B$ and $D$ has the sign of $a$.\\

As a result, $a(z) = \int_0^z \int_0^x \exp{\left(A(\exp{\alpha x}-\exp{\alpha y}) - B (\exp{x}-\exp{y}) -D (x-y)\right) } dx dy$. 
It follows that $\frac{\partial a}{\partial \alpha}(z) = \int_0^z \int_0^x A (x \exp{\alpha x}-y \exp{\alpha y})\exp{\left(A(\exp{\alpha x}-\exp{\alpha y}) - B (\exp{x}-\exp{y}) -D (x-y)\right)} dx dy >0$, so that
if we show that $a(\infty)<\infty$ for $\alpha=1$, it will also hold for $\alpha<1$. By the last bullet above this can only happen if $b < \rho \sigma$, that is $A<B$.\\
With $C=B-A$ we are led to consider the integral
$$\int_0^z \exp{(-C e^x -D x)} \int_0^x \exp{(C e^y +D y)} dy dx $$ 
setting $u = Ce^x, v = Ce^y$ we get a constant times $\int_C^{Ce^z} u^{-D-1} e^{-u} \int_C^u v^{D-1} e^{v} dv du$. By Fubini's theorem this is equal to $\int_C^{Ce^z} v^{D-1} e^{v} \int_v^{Ce^z} u^{-D-1} e^{-u}   du dv$. 
This in turn is less than $\int_C^{Ce^z} v^{D-1} e^{v} \int_v^{\infty} u^{-D-1} e^{-u}   du dv = \int_C^{Ce^z} v^{D-1} e^{v} \Gamma(-D,v) $ where $\Gamma$ is the upper incomplete Gamma function. Now $\Gamma(-D,v) \sim_{v \to \infty} v^{-D-1} e^{-v}$, so that the integrand behaves like $v^{-2}$ at infinity and the integral is finite. \\

All in all, the sole remaining case is $\alpha < 1, b \geq \rho \sigma >0$. Let us operate the change of variable $u = 2 \frac{\rho \sqrt{V_0}}{{\sigma}} \exp{\sigma x}$, we are led to the integral
$$\int_A^Z \exp{(-u + c u^{\alpha})} u^{-\frac{2a}{\sigma^2}-1} \int_A^u \exp{(v - c v^{\alpha})} v^{\frac{2a}{\sigma^2}-1} dv du  $$ with a positive $c$ (by hypothesis $\alpha <1$). We proceed as above: by Fubini's theorem and letting the inner integral go to infinity, this is less than
$$\int_A^Z \exp{(v - c v^{\alpha})} v^{\frac{2a}{\sigma^2}-1}  \int_v^\infty \exp{(-u + c u^{ \alpha})} u^{-\frac{2a}{\sigma^2}-1}   du dv.$$ \\

We claim that $I = \int_v^\infty \exp{(-u + c u^{ \alpha})} u^{-\frac{2a}{\sigma^2}-1}   du \sim \exp{(-v + c v^{ \alpha})} v^{-\frac{2a}{\sigma^2}-1}$, and we can conclude as in the case of the incomplete Gamma function above that $a(\infty)<\infty$. Indeed, by first scaling through the change of variable $u= v z$, $I=\int_1^\infty \exp{(-v z + c v^{ \alpha} z^{ \alpha})} v^{-\frac{2a}{\sigma^2}} z^{-\frac{2a}{\sigma^2}-1} dz$. 
Hence $$I=v^{-\frac{2a}{\sigma^2}} \exp{(-v + c v^{ \alpha})} \int_1^\infty \exp{(-v (z-1) + c v^{\alpha} (z^{ \alpha}-1))} z^{-\frac{2a}{\sigma^2}-1} dz.$$
Setting $z-1=t$ and $r=vt$ we get
$$I=v^{-\frac{2a}{\sigma^2}-1} \exp{(-v + c v^{\alpha})} \int_0^\infty \exp{\left(-r + c v^{ \alpha} \left(\left(1+\frac{r}{v}\right)^{ \alpha}-1 \right)\right)} \left(1+\frac{r}{v}\right)^{\frac{2a}{\sigma^2}-1} dr.$$
As $v \to \infty$ the integrand goes pointwise to $\exp{(-r)}$. Since $\int_0^\infty \exp{(-r)} dr =1$, the last point to check is that we can apply the dominated convergence theorem. This is indeed the case since, on one hand, one always has $(1+\frac{r}{v})^{\frac{2a}{\sigma^2}-1}<(1+r)^{\max{\left(\frac{2a}{\sigma^2}-1, 0\right)}}$ for $v>1$
and on another hand by concavity $c v^{ \alpha} ((1+\frac{r}{v})^{ \alpha}-1)   <  c  v^{ \alpha-1} \alpha r $ with $\alpha<1$, so that for $v$ large enough $  c  v^{ \alpha-1} \alpha<1-\epsilon$ with $\epsilon>0$ and the integrand is less than $e^{-\epsilon r}(1+r)^{\max{\left(\frac{2a}{\sigma^2}-1, 0\right)}}$.\\

We have therefore proven the following result.

\begin{proposition}
$f$ is a martingale if and only if $\alpha\geq 2$, or $\alpha<2$ and either:
\begin{itemize}
\item $\rho \leq 0$
\item $\alpha > 1$
\item	 $\alpha = 1$ and $b \geq \rho \sigma$ 
\end{itemize}
\end{proposition}

\subsubsection{Inversion}
%%%%%%%%%%%%%%%%%%%%%%

Since $(f_t)_{t \geq 0}$ is a true martingale, we can look at the dynamic of $\frac{1}{f}$ under the change of measure induced by the martingale $\frac{f_T}{f_0}$. By Ito's formula, 

$$d\frac{1}{f_t} = -\frac{1}{f^2_t} df_t +\frac{1}{f_t^3} d<f>_t.$$
Now $df_t = \sqrt{V_t} f_t dw_{1,t}$ and $d<f>_t = V_t f_t^2 dt$ so that with $g_t = \frac{1}{f_t}$,
$$dg_t = -\sqrt{V_t} g_t dw_{1,t}+V_t g_t dt.$$
Under the probability $Q = \frac{f_T}{f_0} P$, $\tilde{w}_{1,t} = w_{1,t}-\int_0^t \sqrt{V_s} ds$ is a martingale, and even a Brownian motion by L\'{e}vy's characterization theorem. So
$$dg_t = \sqrt{V_t} g_t (-dw_{1,t}+\sqrt{V_t} dt) = -\sqrt{V_t} g_t d\tilde{w}_{1,t}.$$\\ 

What happens to the variance SDE? Under $Q$, $\tilde{w}_{2,t} = w_{2,t}-\rho \int_0^t \sqrt{V_s} ds$ is a Brownian motion, so that
$$dv_t=(a - b e^{\alpha v_t}) dt + \sigma (dw_{2,t}- \rho \sqrt{V_t} dt) +  \sigma \rho \sqrt{V_t} dt = (a - b e^{\alpha v_t}) dt + \sigma d\tilde w_{2,t} +  \sigma \rho  e^{v_t}dt$$ so it will belong to the same family if and only if $\rho=0$, in which case the inverted model is the initial one, or $\alpha=1$, in which case the mean reversion parameter of the inverted model is given by $b-\rho \sigma$. In particular, if $b=\rho \sigma$, $v_t = a + \sigma \tilde w_{2,t}$ under $Q$.

%%%%%%%%%%%%%%%%%%%%%%%%%%
%\clearpage
%%\input{MorsePotential}
\section{The Hypergeometric Model for $\alpha=1$ and its Morse Potential Representation }
%%%%%%%%%%%%%%%%%%%%%%%%%%%%%%%%%%%%%%%%
In other to price both volatility derivatives \textit{and} equity derivatives we need to further specify the dynamic of the volatility by taking $\alpha=1$ so that we are able to compute the Mellin transform of the forward price which is the essential ingredient to price vanilla options. In that case, the dynamic for the stock and volatility is given by
 
\begin{eqnarray}
df_t&=&f_t e^{ v_t}dw_{1,t},\\
dv_t&=&(a - b e^{v_t}) dt + \sigma dw_{2,t}
\end{eqnarray}

where  as in the general case $dw_{1,t}.dw_{2,t}=\rho dt$.\\ 

We will re-derive some of the results obtained so far. Instead of relating the volatility to the process $Y_t^{\nu}(x)$ we will compute directly the resolvent for the $\alpha$-Hypergeometic model. We owe to the works \cite{pin2010} and \cite{Pintoux2011}, both dealing with interest rate models, the computation strategy used to obtain the key quantities. 
  
\subsection{Volatility Analysis}

We want to compute $\mathbb{E}[e^{\theta v_t}]$. Define a probability $Q$ under which $\tilde{w}_{2,t}= w_{2,t} +\int_0^t\frac{a-be^{v_s}}{\sigma}ds$ is a Brownian motion, we deduce after replacing $\sigma\int_0^te^{v_u}d\tilde{w}_{2,u}=e^{v_t}-e^{v_0}-\frac{\sigma^2}{2}\int_0^te^{v_u}du$ that
\begin{equation*}
\mathbb{E}[e^{\theta v_t}]=e^{-\frac{a}{\sigma^2}v_0 + \frac{b}{\sigma^2}e^{v_0}} e^{-\frac{a^2t}{2\sigma^2}}\mathbb{E}^{Q}\left[   \exp\left(\left(\theta+\frac{a}{\sigma^2}\right)v_t -\frac{b}{\sigma^2}e^{v_t}\right) \exp\left(\beta_1\int_0^t e^{v_u}du -\frac{\beta_2^2}{2}\int_0^te^{2v_u}du\right)  \right]
\end{equation*}

with $\beta_1=\frac{ab}{\sigma^2}+\frac{b}{2}$, $\beta_2^2=\frac{b^2}{\sigma^2}$ and $dv_t=\sigma d\tilde{w}_{2,t}$. Denote by $F(t,v)$ the expectation then it solves, thanks to Feynman-Kac's theorem, the partial differential equation

\begin{align*}
\partial_tF&=\frac{\sigma^{2}}{2}\frac{d^2F}{dv^2}  - \frac{\beta_2^2}{2}e^{2v} F + \beta_1e^{v}F,\\
F(0,v)&=  e^{\left(\theta+\frac{a}{\sigma^2}\right)v -\frac{b}{\sigma^2}e^{v}}.
\end{align*}

Denote by $g(\sigma^2t,v)=F(t,v)$ then is solves the partial differential equation

\begin{align*}
\partial_tg&=- H g,\\
g(0,v)&=e^{\left(\theta+\frac{a}{\sigma^2}\right)v -\frac{b}{\sigma^2}e^{v}}
\end{align*}
with $H = -\frac{1}{2}\frac{d^2}{dv^2}  + \frac{\nu_2^2}{2}e^{2v}  - \nu_1e^{v}$ with $\nu_1=\frac{\beta_1}{\sigma^2}$ and $\nu_2^2=\frac{\beta_2^2}{\sigma^2}$. The operator $H$ involves a Morse potential, see \cite{grosche1988}, page 228 in \cite{gro1998}, \cite{ikeda1999} and the surveys \cite{my1} and \cite{my2}.\\

We denote by $q(t,v,y)$ the heat kernel associated with $e^{-tH}$ then we have
\begin{align*}
F(t,v_0)=\int_{-\infty}^{+\infty}q(\sigma^2 t,v_0,y)F(0,y)dy.
\end{align*}
The Green function associated with the Laplace transform of the heat kernel is given by
\begin{equation}
G(v,y;s^2/2)=\int_0^{+\infty}e^{-\frac{s^2}{2}t}q(t,v,y)dt. \label{GtoQ}
\end{equation}

Taking the Laplace transform of $\mathbb{E}[e^{\theta v_t}]$ we deduce
\begin{align}
\int_0^{+\infty} e^{-\frac{s^2}{2}t}\mathbb{E}\left[e^{\theta v_t} \right]dt&=e^{-\frac{a}{\sigma^2}v_0 + \frac{b}{\sigma^2}e^{v_0}} \int_0^{+\infty}e^{-\left(\frac{a^2}{\sigma^2}+s^2\right)t/2}\int_{-\infty}^{+\infty}q(\sigma^2t,v_0,y)F(0,y)dydt \label{ltFgm}\\
&=\frac{1}{\sigma^2}e^{-\frac{a}{\sigma^2}v_0 + \frac{b}{\sigma^2}e^{v_0}}\int_{-\infty}^{+\infty}  \int_0^{+\infty}e^{-\frac{\eta^2}{2}t}q(t,v_0,y)dt F(0,y)dy\nonumber \\
&=\frac{1}{\sigma^2}e^{-\frac{a}{\sigma^2}v_0 + \frac{b}{\sigma^2}e^{v_0}}\int_{-\infty}^{+\infty} G(v_0,y;\eta^2/2) F(0,y)dy \nonumber
\end{align}
with $\eta^2=\frac{a^2}{\sigma^4}+\frac{s^2}{\sigma^2}$. We know from \cite{my1} pages 341-342 or \cite{my2} page 360 that
\begin{align*}
G(v,y;\eta^2/2)= \frac{\Gamma\left(\eta-\frac{\nu_1}{\nu_2}+\frac{1}{2}\right)}{\nu_2\Gamma(1+2\eta)} e^{-(v+y)/2}W_{\frac{\nu_1}{\nu_2},\eta}\left( 2\nu_2e^{y_{>}}\right)M_{\frac{\nu_1}{\nu_2},\eta}\left( 2\nu_2e^{y_{<}}\right)
\end{align*} 
with $y_{>}=\max(v,y)$ and $y_{<}=\min(v,y)$  while $W_{\kappa,\eta}$ and $M_{\kappa,\eta}$ are the Whittaker functions related to the confluent hypergeometric functions by the relations
\begin{align*}
W_{\kappa,\eta}(z)&=z^{\eta+\frac{1}{2}}e^{-z/2}\Psi\left(\eta-\kappa +\frac{1}{2},1+2\eta;z \right),\\
M_{\kappa,\eta}(z)&=z^{\eta+\frac{1}{2}}e^{-z/2}\Phi\left(\eta-\kappa +\frac{1}{2},1+2\eta;z \right).
\end{align*}
It is known that the heat kernel is given by 
\begin{align}
q(t,v,y)&=\int_0^{+\infty}\frac{e^{2\frac{\nu_1}{\nu_2}u}}{2\sinh(u)}e^{-\frac{\nu_1}{\nu_2}(e^v+e^y)\coth(u)}\theta\left(2\frac{\nu_1}{\nu_2}e^{(v+y)/2}/\sinh(u),t\right)du, \label{heatKernelMorsePotential}\\
\theta(r,t)&=\frac{r}{(2\pi^3t)^{\frac{1}{2}}} e^{\pi^2/(2t)}\int_0^{+\infty}e^{-u^2/(2t)}e^{-r\cosh(u)}\sinh(u)\sin\left(\frac{u\pi}{t}\right)du. \label{thetaFunc}
\end{align}
We wish to compute 
\begin{align}
\int_{-\infty}^{+\infty} G(v_0,y;\eta^2/2) F(0,y)dy&=\int_{-\infty}^{v_0}G(v_0,y;\eta^2/2) F(0,y)dy + \int_{v_0}^{+\infty}G(v_0,y;\eta^2/2) F(0,y)dy \nonumber  \\
&=J_1+J_2 .\label{J1J2}
\end{align}

We have
\begin{align*}
J_1&=\frac{\Gamma\left(\eta-\frac{\nu_1}{\nu_2}+\frac{1}{2}\right)}{\nu_2\Gamma(1+2\eta)}e^{-v_0/2} W_{\frac{\nu_1}{\nu_2},\eta}\left( 2\nu_2e^{v_0}\right) \int_{-\infty}^{v_0} e^{-y/2}M_{\frac{\nu_1}{\nu_2},\eta}\left( 2\nu_2e^{y}\right)F(0,y)dy\\
&=\frac{\Gamma\left(\eta-\frac{\nu_1}{\nu_2}+\frac{1}{2}\right)}{\nu_2\Gamma(1+2\eta)}e^{-v_0/2} W_{\frac{\nu_1}{\nu_2},\eta}\left( 2\nu_2e^{v_0}\right)(2\nu_2)^{\frac{1}{2}-n-\frac{a}{\sigma^2}}\int_{0}^{z_0}z^{\eta-1+\theta+\frac{a}{\sigma^2}}e^{-z}\Phi\left(\eta-\frac{a}{\sigma^2},1+2\eta;z\right)dz
\end{align*}
where $z_0=2\nu_2e^{v_0}$, $\frac{\nu_1}{\nu_2}=\frac{a}{\sigma^2}+\frac{1}{2}$ and we used the representation for the Whittaker function $M_{\kappa,\eta}(z)$. Similarly, the representation for the Whittaker function $W_{\kappa,\eta}(z)$ leads to
\begin{align*}
J_2&= \frac{\Gamma\left(\eta-\frac{\nu_1}{\nu_2}+\frac{1}{2}\right)}{\nu_2\Gamma(1+2\eta)}e^{-v_0/2} M_{\frac{\nu_1}{\nu_2},\eta}\left( 2\nu_2e^{v_0}\right) \int_{v_0}^{+\infty} e^{-y/2}W_{\frac{\nu_1}{\nu_2},\eta}\left( 2\nu_2e^{y}\right)F(0,y)dy\\
%\int_{v_0}^{+\infty} e^{-y/2}W_{\frac{\nu_1}{\nu_2},\eta}\left( 2\nu_2e^{y}\right)f(0,v)dv = (2\nu_2)^{\frac{1}{2}-n-\frac{a}{\sigma^2}}\int_{z_0}^{+\infty}z^{\eta-1+n+\frac{a}{\sigma^2}}e^{-z}\Psi\left(\eta-\frac{a}{\sigma^2},1+2\eta;z\right)dz\\
&= \frac{\Gamma\left(\eta-\frac{\nu_1}{\nu_2}+\frac{1}{2}\right)}{\nu_2\Gamma(1+2\eta)}e^{-v_0/2} M_{\frac{\nu_1}{\nu_2},\eta}\left( 2\nu_2e^{v_0}\right) (2\nu_2)^{\frac{1}{2}-n-\frac{a}{\sigma^2}}\int_{z_0}^{+\infty}z^{\eta-1+ \theta +\frac{a}{\sigma^2}}e^{-z}\Psi\left(\eta-\frac{a}{\sigma^2},1+2\eta;z\right)dz.
\end{align*}

To connect these results to the previous ones we just need 
\begin{remark}\label{integToI}
If we denote $a_1-1=\eta -\frac{a}{\sigma^2}$ and $b_1=1+2\eta$ then the two integrals above (i.e. involved in $J_1$ and $J_2$) can be rewritten as
\begin{align*}
\int_{z_0}^{+\infty}z^{b_1-a_1 + \theta -1}e^{-z}\Psi\left(a_1-1,b_1;z\right)dz,\\
\int_{0}^{z_0}z^{b_1-a_1 + \theta -1}e^{-z}\Phi\left(a_1-1,b_1;z\right)dz,
\end{align*}
which can be computed thanks to the expressions obtained for $I_1$ and $I_2$.
\end{remark}

\subsubsection{The variance swaps revisited}

The variance swap is given by

\begin{equation*}
t\textsc{vs}(t)=\int_0^{t}\mathbb{E}\left[e^{2v_u} \right]du
\end{equation*}
and its Laplace transform is
\begin{equation}
\int_{0}^{+\infty}e^{-s^2t/2} t\textsc{vs}(t) dt = \frac{2}{s^2}\int_{0}^{+\infty}e^{-s^2t/2} \mathbb{E}\left[e^{2v_t} \right]dt. \label{ltVs} 
\end{equation}

The equation \eqref{ltVs} is the left hand side of equation \eqref{ltFgm} and leads to the integrals $J_1$ and $J_2$ given above and thanks to Remark \ref{integToI} the series representations for $I_1$ and $I_2$ enable an efficient computation of the variance swap.

\begin{remark}
To check that  
\begin{align*}
\int_0^{+\infty} e^{-\frac{s^2}{2}t}\mathbb{E}\left[e^{\theta v_t} \right]dt < +\infty
\end{align*}
we need to verify that 
\begin{align*}
\int_{v_0}^{+\infty}e^{-y/2} W_{\frac{\nu_1}{\nu_2},\eta}\left( 2\nu_2e^{y}\right)\exp\left\lbrace\left(\theta +\frac{a}{\sigma^2}\right)y - \frac{b}{\sigma^2}e^y \right\rbrace dy,\\
\int_{-\infty}^{v_0}e^{-y/2} M_{\frac{\nu_1}{\nu_2},\eta}\left( 2\nu_2e^{y}\right)\exp\left\lbrace\left(\theta +\frac{a}{\sigma^2}\right)y - \frac{b}{\sigma^2}e^y \right\rbrace dy
\end{align*}
are finite. As the Whittaker $W_{\kappa,\eta}$ function is related to the confluent  hypergeometric function $\Psi$ and using relation 6.2.2 of \cite{bea2010}, which is
\begin{align*}
\Psi(\alpha,\beta;z)\sim z^{-\alpha} \quad \textrm{if} \quad \Re(z)\rightarrow +\infty  \textrm{ and } \Re(\alpha) >0,
\end{align*}
we conclude that because $\Re(\nu_2)>0$ and $\Re\left(\eta-\frac{\nu_1}{\nu_2} +\frac{1}{2}\right)>0 $ for $s$ large enough ($\eta$ depends on $s$) so the integrand of the first integral behaves like
\begin{align*}
e^{y\left(\frac{2a}{\sigma^2} + \theta \right)}\exp\left\lbrace \frac{-2b}{\sigma^2}e^{y} \right\rbrace  \quad \textrm{as} \quad y\rightarrow +\infty
\end{align*}
and therefore the integral will be finite for all values of $\theta $ ($b$ is positive). For the second  integral  replacing the Whittaker function $M_{\kappa,\eta}$ by its expression and using the property 13.2.13 of \cite{DLMF}, that is $\Phi(\alpha,\beta;z)\sim 1 \; \textrm{for} \; z\sim 0 $, we deduce that the integrand behaves like 
\begin{equation*}
e^{y\left(\eta + \theta + \frac{a}{\sigma^2} \right)} \quad \textrm{as} \quad y\rightarrow -\infty,
\end{equation*}
for all values of $\theta$ there exists a value for $s$ such that $\eta + \theta + \frac{a}{\sigma^2}>0$ so the integral is finite. Notice also that to the extent that $\Re(\nu_2)>0$ we can have a potential with complex coefficients and the integrals will remain finite.
\end{remark}

\subsection{The Mellin Transform of the Spot}

In order to perform the pricing of vanilla options we need to compute the Mellin transform of the spot. We have 
\begin{eqnarray*}
\mathbb{E}\left[\left(\frac{f_t}{f_0}\right)^\lambda \right]  &=&\mathbb{E}\left[ \exp\left(-\frac{\lambda}{2}\int_0^te^{2v_u}du + \lambda \int_0^te^{2v_u}dw_{1,u}\right)\right] \\
&=&\mathbb{E}\left[ \exp\left(-\frac{\lambda}{2}\int_0^te^{2v_u}du +  \lambda\rho \int_0^te^{2v_u}dw_{2,u} + \lambda\sqrt{1-\rho^2} \int_0^te^{2v_u}dw_{2,u}^{\bot}\right)\right] \\
&=&\mathbb{E}\left[ \exp\left(\left( -\frac{\lambda}{2} + \frac{\lambda^2(1-\rho^2)}{2}\right)\int_0^te^{2v_u}du  +    \lambda\rho \int_0^te^{2v_u}dw_{2,u} \right)\right]
\end{eqnarray*}

where we used the standard Brownian motion $(w_{2,t},w_{2,t}^\bot)_{t \geq 0}$. Furthermore, the relation 

\begin{equation}
\sigma \int_0^te^{v_u}dw_{2,u} = e^{v_t}-e^{v_0}-\int_0^te^{v_u}(a-be^{v_u})du -\frac{\sigma^2}{2}\int_0^t e^{v_u}du \label{intevdw}
\end{equation} 
leads to
\begin{align*}
\mathbb{E}\left[\left(\frac{f_t}{f_0}\right)^\lambda \right]= e^{-\frac{\lambda \rho}{\sigma}e^{v_0}} \mathbb{E}\left[ \exp\left( \alpha_0 e^{v_t} + \alpha_1 \int_0^te^{v_s}ds  - \frac{\alpha_2^2}{2} \int_0^te^{2v_s}ds \right)  \right]
\end{align*}
with 
\begin{align*}
\alpha_0&=\frac{\lambda\rho}{\sigma},\\
\alpha_1&=-\frac{\lambda\rho}{\sigma}\left(a+ \frac{\sigma^2}{2}\right),\\
\alpha_2^2&=-\lambda^2(1-\rho^2) -\frac{2b\rho \lambda}{\sigma } + \lambda.
\end{align*}

Using Girsanov's theorem we deduce that
\begin{align*}
J&=\mathbb{E}\left[ \exp\left( \alpha_0 e^{v_t} + \alpha_1 \int_0^te^{v_s}ds  - \frac{\alpha_2^2}{2} \int_0^te^{2v_s}ds \right)  \right]\\
&=\mathbb{E}^{Q}\left[ \exp\left( \alpha_0 e^{v_t} + \alpha_1 \int_0^te^{v_s}ds  - \frac{\alpha_2^2}{2} \int_0^te^{2v_s}ds \right) \exp\left( \int_0^t\frac{a-be^{v_u}}{\sigma}d\tilde w_s  - \frac{1}{2} \int_0^t\frac{(a-be^{v_u})^2}{\sigma^2}ds \right)  \right]
\end{align*}

with $dv_t=\sigma d\tilde w_t$ and $ \tilde  w_t = w_{2,t} +\int_0^t\frac{a-be^{v_u}}{\sigma}du $ a Brownian motion under $Q$. Using again the equality \eqref{intevdw} (with  convenient parameters) we deduce that
\begin{align}
\mathbb{E}\left[\left(\frac{f_t}{f_0}\right)^\lambda \right]  &=e^{-\frac{a}{\sigma^2}v_0 + (\frac{b}{\sigma^2}-\frac{\lambda \rho}{\sigma}) e^{v_0}} e^{-\frac{a^2t}{2\sigma^2}}\mathbb{E}^Q\left[ \exp\left(\frac{av_t}{\sigma^2}+\beta_0 e^{v_t}+ \beta_1\int_0^t e^{v_s}ds -\frac{\beta_2^2}{2}\int_0^te^{2v_s}ds\right)  \right] \label{Mellin1}
\end{align}
with 
\begin{align*}
\beta_0 &= \alpha_0 - \frac{b}{\sigma^2} = \frac{\lambda \rho \sigma-b}{\sigma^2},\\
\beta_1 &=  \alpha_1 + b\left( \frac{a}{\sigma^2}+ \frac{b}{2}\right) =  (b-\lambda\rho \sigma) \left(\frac{a}{\sigma^2} + \frac{1}{2}\right),\\
\beta_2^2 &= \alpha_2^2 + \frac{b^2}{\sigma^2} =  -\lambda^2(1-\rho^2) +  \lambda \left(1 -  \frac{2b\rho}{\sigma}\right)+\frac{b^2}{\sigma^2}.
\end{align*}
As above, introduce
$$\nu_1=\frac{\beta_1}{\sigma^2}, \quad \nu_2^2=\frac{\beta_2^2}{\sigma^2} $$
and $F(0,v)=\exp\left(\frac{av}{\sigma^2}+\beta_0 e^{v}\right)$. We denote by $F(t,v)$ the expectation in \eqref{Mellin1}, then thanks to Feynman-Kac's formula it solves the partial differential equation
\begin{align*}
\partial_tF&=\frac{\sigma^{2}}{2}\frac{d^2F}{dv^2}  - \frac{\beta_2^2}{2}e^{2v} F + \beta_1e^{v}F,\\
F(0,v)&=  e^{\frac{av}{\sigma^2} + \beta_0e^{v}}.
\end{align*}
Proceeding as above we obtain the following integral representation
\begin{equation*}
F(t,v_0)=\int_{-\infty}^{+\infty}q(\sigma^2 t,v_0,y)F(0,y)dy,
\end{equation*}

which requires the kernel $q$, known from \eqref{heatKernelMorsePotential}, but is hard to exploit. We can also use the Green function given by \eqref{GtoQ} as follows, we compute the Laplace transform
\begin{align}
\int_0^{+\infty} e^{-\frac{s^2}{2}t}  e^{-\frac{a^2t}{2\sigma^2}} F(t,v_0)dt &=\frac{1}{\sigma^2}\int_{-\infty}^{+\infty} G(v_0,y;\eta^2/2) F(0,y)dy \label{LtF}
\end{align}
with $\eta^2=\frac{a^2}{\sigma^4}+\frac{s^2}{\sigma^2}$ and proceed as in the previous example and write the integral appearing in the r.h.s of \eqref{LtF} as in \eqref{J1J2} (as a sum of two integrals denoted $J_1$ and $J_2$ given below). Taking into account the particular form of $F(0,v)$ we are led to the computation of
\begin{align}
J_1&=\frac{\Gamma\left(\eta-\frac{\nu_1}{\nu_2}+\frac{1}{2}\right)}{\nu_2\Gamma(1+2\eta)}e^{-v_0/2} W_{\frac{\nu_1}{\nu_2},\eta}\left( 2\nu_2e^{v_0}\right) \int_{-\infty}^{v_0} e^{-y/2}M_{\frac{\nu_1}{\nu_2},\eta}\left( 2\nu_2e^{y}\right)F(0,y)dy \label{J1Spot}, \\
J_2&= \frac{\Gamma\left(\eta-\frac{\nu_1}{\nu_2}+\frac{1}{2}\right)}{\nu_2\Gamma(1+2\eta)}e^{-v_0/2} M_{\frac{\nu_1}{\nu_2},\eta}\left( 2\nu_2e^{v_0}\right) \int_{v_0}^{+\infty} e^{-y/2}W_{\frac{\nu_1}{\nu_2},\eta}\left( 2\nu_2e^{y}\right)F(0,y)dy \label{J2Spot}
\end{align}
where $z_0=2\nu_2e^{v_0}$. Using the representation for the Whittaker functions $W_{\kappa,\eta}(z)$ and $M_{\kappa,\eta}(z)$ the two integrals above can be transformed into 

\begin{equation*}
%\int_{-\infty}^{v_0} e^{-y/2}M_{\frac{\nu_1}{\nu_2},\eta}\left( 2\nu_2e^{y}\right)F(0,y)dy = (2 \nu_2)^{\frac{1}{2}-\frac{a}{\sigma^2}}\int_{0}^{z_0}z^{\eta-1+\frac{a}{\sigma^2}}e^{\left(-\frac{1}{2}+ \frac{\beta_0}{2\nu_2}\right) z}\Phi\left(\eta-\frac{\nu_1}{\nu_2}+\frac{1}{2},1+2\eta;z\right)dz 
\int_{-\infty}^{v_0} e^{-y/2}M_{\frac{\nu_1}{\nu_2},\eta}\left( 2\nu_2e^{y}\right)F(0,y)dy = (2 \nu_2)^{\frac{1}{2}-\frac{a}{\sigma^2}}I_1
\end{equation*}
and

\begin{equation*}
%\int_{v_0}^\infty e^{-y/2}W_{\frac{\nu_1}{\nu_2},\eta}\left( 2\nu_2e^{y}\right)F(0,y)dy = (2 \nu_2)^{\frac{1}{2}-\frac{a}{\sigma^2}}\int_{z_0}^\infty z^{\eta-1+\frac{a}{\sigma^2}}e^{\left(-\frac{1}{2}+ \frac{\beta_0}{2\nu_2}\right) z}\Psi\left(\eta-\frac{\nu_1}{\nu_2}+\frac{1}{2},1+2\eta;z\right)dz\\
\int_{v_0}^\infty e^{-y/2}W_{\frac{\nu_1}{\nu_2},\eta}\left( 2\nu_2e^{y}\right)F(0,y)dy = (2 \nu_2)^{\frac{1}{2}-\frac{a}{\sigma^2}}I_2
\end{equation*}

with $z_0=2 \nu_2 e^{v_0}$,
\begin{equation}
I_1=\int_{0}^{z_0}z^{\eta-1+\frac{a}{\sigma^2}}e^{(-\frac{1}{2}+\frac{\beta_0}{2\nu_2})z}\Phi\left(\eta-\frac{\nu_1}{\nu_2}+\frac{1}{2},1+2\eta;z\right)dz \label{I1Spot}
\end{equation}
  and 
\begin{equation}
I_2=\int_{z_0}^\infty z^{\eta-1+\frac{a}{\sigma^2}}e^{(-\frac{1}{2}+\frac{\beta_0}{2\nu_2})z}\Psi\left(\eta-\frac{\nu_1}{\nu_2}+\frac{1}{2},1+2\eta;z\right)dz. \label{I2Spot}
\end{equation}

The behaviour of these integrals will be driven by the quantity $-\frac{1}{2}+ \frac{\beta_0}{2\nu_2}$, so let us investigate it.

\begin{lemma} \label{delta}
Let $\delta(\lambda) = -\frac{1}{2}+ \frac{\beta_0(\lambda)}{2\nu_2(\lambda)}$. Denote by $\lambda_{-}, \lambda_{+}$ the roots of the polynomial $(\lambda \rho \sigma -b)^2 +\sigma^2 \lambda (1-\lambda)$. Then:
\begin{itemize}
\item $\delta(\lambda)$ is defined for $\lambda \in ]\lambda_{-}, \lambda_+[$, with $\lambda_{-}< 0<1<\lambda_+$
\item $\delta(0) = -1$
\item $\delta(1)=0$ if  $b<\rho \sigma$, $\delta(1)=-1$ if $b>\rho \sigma$
\item $\delta(\lambda)<0$ for $\lambda \in ]0,1[$
\item When $\rho<0$, $\delta(\lambda)<0$ for $\lambda \in ]\lambda_{-}, \lambda_+[$
\end{itemize}
\end{lemma}

\begin{proof}
Observe that
$$\sigma^2 \beta_2^2 = (\lambda \rho \sigma -b)^2 + \sigma^2 \lambda (1-\lambda) $$
so that with $\nu_2 = \frac{\beta_2}{\sigma} $
$$\frac{\beta_0}{2 \nu_2}= \frac{1}{2} \frac{(\lambda \rho \sigma -b)}{\sqrt{(\lambda \rho \sigma -b)^2 +\sigma^2 \lambda (1-\lambda)}}.$$ 
In particular, $\frac{\beta_0}{2 \nu_2}(\lambda=0)=-\frac{1}{2}$ and $\frac{\beta_0}{2 \nu_2}(\lambda=1)=sgn(\rho \sigma-b)\frac{1}{2}$.
Note that $-\frac{1}{2}+ \frac{\beta_0}{2\nu_2} = \frac{1}{2} (\frac{(\lambda \rho \sigma -b)}{\sqrt{(\lambda \rho \sigma -b)^2 +\sigma^2 \lambda (1-\lambda)}}-1)<0$
as soon as $\lambda \in [0,1[$ or $\lambda=1$ and $b>\rho \sigma$.
If $\rho<0$ then $\frac{\beta_0}{2 \nu_2}(\lambda=1)=-\frac{1}{2}$  and the maximum $m$ of $\lambda \to \frac{\beta_0}{2 \nu_2}(\lambda)$ is attained between $0$ and $1$ with $-\frac{1}{2}<m<0$. Moreover, $\beta_2$ is well defined as long as $\lambda_-\leq\lambda\leq\lambda_+$ with $\lambda_-<0<1<\lambda_+$. The last constraint to check is $\beta_0<0$.	 Assuming $\rho<0$, this amounts to $\lambda>\frac{b}{\rho \sigma}$ which is negative and even smaller than $\lambda_-$ since it cancels the first squared monomial in the expression of $\beta_2^2$.
It follows that all the range $\lambda_-, \lambda_+$ is allowed, with the exponent $\frac{\beta_0}{2 \nu_2}$ living between $-\frac{1}{2}$ and its maximum $m$ for $\lambda \in [0,1]$ and decreasing to $-\infty$ close to $\lambda_-$ or $\lambda_+$.

\end{proof}

\paragraph{Computation of $I_1$:}
%%%%%%%%%%%%%%%%%%%%%%%%%%%%%%%%
Because $\Phi\left(\eta-\frac{\nu_1}{\nu_2}+\frac{1}{2},1+2\eta;0\right)=1$, $I_1$ is well defined if and only if
$\eta+\frac{a}{\sigma^2}>0$, which is always true since $\eta=\sqrt{\frac{a^2}{\sigma^4}+\frac{s^2}{\sigma^2}}$.
By Fubini's theorem, $I_1=\sum_{n=0}^\infty \frac{(\eta-\frac{\nu_1}{\nu_2}+\frac{1}{2})_n}{(1+2\eta)_n n!} \int_0^{z_0} z^{\eta-1+\frac{a}{\sigma^2}+n}e^{\delta(\lambda) z} dz $.
Let $i_n = \int_0^{z_0} z^{\eta-1+\frac{a}{\sigma^2}+n}e^{\delta(\lambda)z} dz$, the following results focus on the determination of this key quantity.
\subparagraph{When $\delta(\lambda)<0 $:}
 Then $i_n=(-\delta(\lambda))^{-\eta-\frac{a}{\sigma^2}-n} \gamma\left(\eta+\frac{a}{\sigma^2}+n, -\delta(\lambda) z_0\right) $ with $\gamma$ the lower incomplete Gamma function. By integration by parts we have 
$$\delta(\lambda) i_{n+1}= z_0^{\eta+\frac{a}{\sigma^2}+n} e^{\delta(\lambda)z_0}-\left(\eta+\frac{a}{\sigma^2}+n\right) i_n.$$ 
So there is a straightforward recurrence to compute the term of the series of $I_1$.

\subparagraph{When $\delta(\lambda)=0 $:}
Then $i_n = \frac{z_0^{\eta+\frac{a}{\sigma^2}+n}}{\eta+\frac{a}{\sigma^2}+n}$.

\paragraph{Computation of $I_2$:}
%%%%%%%%%%%%%%%%%%%%%%%%%%%%%%%%

Let us investigate first the key coefficients in $I_2$.

Note that
$$\frac{\nu_1}{\nu_2} = -\left(\frac{a}{\sigma^2} + \frac{1}{2}\right)  2 \frac{\beta_0}{2 \nu_2},$$

and we have
\begin{lemma}
$I_2(\lambda=1)$ is finite. 
\end{lemma}
\begin{proof}
We know that $\Psi\left(\eta-\frac{\nu_1}{\nu_2}+\frac{1}{2},1+2\eta;z\right) \sim z^{\frac{\nu_1}{\nu_2}-\eta-\frac{1}{2}}$, so the integrand behaves like 
$$z^{\frac{\nu_1}{\nu_2}+\frac{a}{\sigma^2}-\frac{3}{2}} e^{\frac{(-1+sgn(\rho \sigma-b))z}{2}} $$ at infinity.
Therefore $I_2$ will be finite if $b\geq\rho \sigma$. If $b<\rho \sigma$ then it will be finite if and only if
$\frac{\nu_1}{\nu_2}+\frac{a}{\sigma^2}<\frac{1}{2}$ which is true since $\frac{\nu_1}{\nu_2}=- \left(\frac{a}{\sigma^2} + \frac{1}{2}\right)$.
\end{proof}\\

\subparagraph{When $\delta(\lambda)<-1 $ (in particular, when $\lambda \in ]1, \lambda_+[$):} We have
$$z^{\eta-1+\frac{a}{\sigma^2}}e^{\delta(\lambda)z}\Psi\left(\eta-\frac{\nu_1}{\nu_2}+\frac{1}{2},1+2\eta;z\right) = \frac{1}{2i\pi} \int_{-i\infty}^{i\infty} \frac{\Gamma(bb-1+t) \Gamma(t)}{\Gamma(aa-1+t)} z^{-\eta-1+\frac{a}{\sigma^2}-t} e^{(\delta(\lambda)+1) z} dt $$
with $aa = \eta-\frac{\nu_1}{\nu_2}+\frac{3}{2} $ and $bb=1+2\eta$.\\

Moreover, we know that the integral converges locally uniformly in $z$, so that we can apply Fubini's theorem and obtain
$$I_2 = \frac{1}{2i\pi} \int_{-i\infty}^{i\infty} \frac{\Gamma(bb-1+t) \Gamma(t)}{\Gamma(aa-1+t)} \int_{z_0}^\infty
z^{-\eta-1+\frac{a}{\sigma^2}-t} e^{(\delta(\lambda)+1)z} dz dt  $$ 
where the inner integral is finite since $\delta(\lambda)+1<0$ by assumption.

Now, define 
\begin{equation*}
j(t)=\int_{z_0}^\infty z^{-\eta-1+\frac{a}{\sigma^2}-t} e^{(\frac{1}{2}+\frac{\beta_0}{2\nu_2})z} dz = |\frac{1}{2}+\frac{\beta_0}{2\nu_2}|^{\eta-\frac{a}{\sigma^2}+t} \Gamma(-\eta+\frac{a}{\sigma^2}-t, \frac{z_0}{2} | 1+\frac{\beta_0}{\nu_2}| )
\end{equation*}
where $\Gamma(,)$ denotes the upper incomplete Gamma function. Since $z_0 \ne 0$ we know (\cite{DLMF} 8.2, (ii)) that $\Gamma(a,z)$ is an entire function of $a$, and will not contribute to the poles of the integrand.

\subparagraph{An explicit hypergeometric series for $I_2$:} We apply the theorem of residues to get a series from the above complex integral.
The poles of the integrand are located:
\begin{itemize}
\item at $t=-n, n \in \mathbb{N}$, with residue $\frac{\Gamma(bb-1-n)}{\Gamma(aa-1-n) n!} (-1)^n  j(-n)$
\item at $t=-n+1-bb, n \in \mathbb{N}$, with residue $\frac{\Gamma(-n+1-bb)}{\Gamma(aa-n-bb) n!} (-1)^n j(-n+1-bb)  $
\end{itemize}
so that by Cauchy's residue theorem we get 
$$I_2 =\sum_{n=0}^\infty \frac{(-1)^{n}}{n!} (\frac{\Gamma(bb-1-n) j(-n)}{\Gamma(aa-1-n)} + \frac{ \Gamma(1-bb-n) j(-n+1-bb)}{\Gamma(aa-bb-n)})  $$
$I_2$ can be computed easily by writing down the explicit  recurrences between successive terms of the two series. Calls to the $\Gamma$ and $\Gamma(,)$ functions are only required for the constant and index zero terms.
We have the following recurrence relation for a generic $j^*(t)=\int_{z_0}^\infty z^{\alpha-t} e^{-\beta z} dz$:
$$j^*(-(n+1)) = \frac{e^{-\beta z_0} z_0^{\alpha+n+1} }{\beta}+\frac{\alpha+n+1}{\beta} j^*(-n).$$

\paragraph{Complete expression of the double transform:}
%%%%%%%%%%%%%%%%%%%%%%%%%%%%%%%%%%%%%%%%%%%%%%%%%%%%%%

We now have all the elements to compute the Laplace transform of the asset. In fact, we have
$$\int_0^\infty e^{-\frac{s^2}{2} t} \mathbb{E}\left[\left(\frac{f_t}{f_0}\right)^\lambda\right] dt = \frac{1}{\sigma^2} e^{-\frac{a}{\sigma^2} v_0 +(\frac{b}{\sigma^2}-\frac{\lambda \rho}{\sigma}) e^{v_0}} \int_{-\infty}^{\infty} G(v_0,y, \frac{\eta^2}{2}) F(0,y) dy$$

where $F(0,y)=e^{\frac{a}{\sigma^2} y + \beta_0 e^{y}}$, $\eta^2=\frac{a^2}{\sigma^4}+\frac{s^2}{\sigma^2}$ and

$$ G(v,y, \frac{\eta^2}{2}) = \frac{ 2^{2\eta+1} \nu_2^{2 \eta} \Gamma\left(\eta-\frac{\nu_1}{\nu_2}+\frac{1}{2}\right)}{\Gamma(1+2\eta)} e^{\eta (v+y)}
e^{-\nu_2 (e^v+e^y)} [1(y>v) \Psi(;2\nu_2 e^y) \Phi(;2\nu_2 e^v)+1(y<v) \Psi(;2\nu_2 e^v) \Phi(;2\nu_2 e^y)]$$
where the 1st and 2nd arguments of the Kummer functions are $\eta-\frac{\nu_1}{\nu_2}+\frac{1}{2}, 1+2 \eta$. This leads to
\begin{equation}
\int_0^\infty e^{-\frac{s^2}{2} t} \mathbb{E}\left[\left(\frac{f_t}{f_0}\right)^\lambda\right]dt =  
\frac{ 2^{\eta+1-\frac{a}{\sigma^2}} \nu_2^{2\eta-\frac{a}{\sigma^2}} \Gamma\left(\eta-\frac{\nu_1}{\nu_2}+\frac{1}{2}\right)}{\sigma^2 \Gamma(1+2\eta)} e^{(\eta-\frac{a}{\sigma^2}) v_0}e^{(\frac{b}{\sigma^2}-\frac{\lambda \rho}{\sigma}-\nu_2) e^{v_0}} [\Phi(;z_0) I_2 + \Psi(;z_0) I_1]. \label{Ltflambda} 
\end{equation}

\subsubsection{A purely analytical approach to the martingale property}
%%%%%%%%%%%%%%%%%%%%%%%%%%%%%%%%%%%%%

Since $(f_t)_{t \geq 0}$ is a positive super martingale, $\mathbb{E}\left[\left(\frac{f_t}{f_0}\right)\right] \leq 1$ and it follows that $(f_t)_{t \geq 0}$ is a martingale if and only if for some $s>0$
$$\int_0^\infty e^{-\frac{s^2}{2} t} \mathbb{E}\left[\left(\frac{f_t}{f_0}\right)^\lambda\right] dt = \frac{2}{s^2}$$
with $\lambda$ set to 1. Therefore there is some hope to re-find the results of the "full-blown" section (in case $\alpha=1$) in the previous calculations when $\lambda=1$. Note also that when $\lambda=0$ this should hold irrespective of the model parameters.\\

We start with the following two useful lemmas.

\begin{lemma} \label{nu12}
If $\lambda=0$, then
\begin{itemize}
\item $\frac{\nu_1}{\nu_2}=\frac{a}{\sigma^2}+\frac{1}{2}$
\item $\frac{b}{\sigma^2}-\frac{\lambda \rho}{\sigma}-\nu_2=0$
\end{itemize}
This also holds when $\lambda=1$ if and only if $b\geq \rho \sigma$.
\end{lemma}

The other useful result is the next lemma.
\begin{lemma} \label{phiPsi}
Under the conditions of Lemma \ref{nu12},
\begin{itemize}
\item $I_1 = \frac{z_0^{\eta+\frac{a}{\sigma^2}}}{\eta+\frac{a}{\sigma^2}} e^{-z_0} \Phi(\eta-\frac{a}{\sigma^2}+1,1+2\eta;z_0)$
\item $I_2 = z_0^{\eta+\frac{a}{\sigma^2}} e^{-z_0} \Psi(\eta-\frac{a}{\sigma^2}+1,1+2\eta;z_0)$.
\end{itemize}

\end{lemma}

\begin{proof}
Combining lemmas \ref{delta} and \ref{nu12} we get that $I_1$ writes
$$I_1=\int_{0}^{z_0}z^{\eta-1+\frac{a}{\sigma^2}}e^{-z}\Phi\left(\eta-\frac{a}{\sigma^2},1+2\eta;z\right)dz.$$
With $aa=\eta-\frac{a}{\sigma^2}+1, bb=1+2\eta$ the integrand is $z^{bb-aa-1}e^{-z}\Phi\left(aa-1,bb;z\right)$ and by (\cite{DLMF}, 13.3.19) this is also 
$$\frac{1}{bb-aa}\frac{d}{dz} z^{bb-aa}e^{-z}\Phi\left(aa,bb;z\right).$$
Since $bb>aa$ and $\Phi\left(aa,bb;0\right)=1$ the first assertion follows. $I_2$ rewrites 
$$I_2=\int_{z_0}^{\infty}z^{\eta-1+\frac{a}{\sigma^2}}e^{-z}\Psi\left(\eta-\frac{a}{\sigma^2},1+2\eta;z\right)dz,$$ and by (\cite{DLMF}, 13.3.26)
this is also
$$-\frac{d}{dz} z^{bb-aa}e^{-z}\Psi\left(aa,bb;z\right).$$
Since $\Psi\left(aa,bb;z\right) \sim z^{-aa}$ as $z \to \infty$ the result follows.
\end{proof}\\

To conclude, let us first assemble the pieces together:\\
\paragraph{In case $\lambda=0$ or $\lambda=1$ with $b \geq \rho \sigma$:}
Combining the expression \eqref{Ltflambda} with lemmas \ref{nu12} and \ref{phiPsi} and using $z_0=2 \nu_2 e^{v_0}$we get, with
$A(z_0) = \int_0^\infty e^{-\frac{s^2}{2} t} \mathbb{E}\left[\left(\frac{f_t}{f_0}\right)^\lambda\right] dt$,

$$A(z_0)= \frac{ 2 \Gamma\left(\eta-\frac{a}{\sigma^2}\right)}{\sigma^2 \Gamma(1+2\eta)}
 z_0^{2\eta} e^{-z_0} [\Phi(;z_0) \Psi(\eta-\frac{a}{\sigma^2}+1,1+2\eta;z_0)  + \Psi(;z_0) \frac{\Phi(\eta-\frac{a}{\sigma^2}+1,1+2\eta;z_0) }{\eta+\frac{a}{\sigma^2}}  ].$$

We know that in the case $\lambda=0$, this expression is the Laplace transform with respect to time of 1. Since it is 
equal to the Laplace transform with respect to time of $\mathbb{E}\left[\left(\frac{f_t}{f_0}\right)\right]$ when $\lambda=1$ with $b \geq \rho \sigma$
we have proven the martingale property in that case. 

\paragraph{Working out the identity  with Kummer functions:} We know that this expression should be equal to $\frac{2}{s^2}$ for any $s$.  Can we show this?

\subparagraph{When $z_0 \to 0$:}
Then $\Phi(;z_0) \to 1$ and at least when $\eta>\frac{1}{2}$, $\Psi(aa, bb,;z) \sim \frac{\Gamma(bb-1)}{\Gamma(aa)} z^{1-bb}$ with $bb=1+2\eta$ so that
$$ A(z_0) \sim \frac{ 2 \Gamma\left(\eta-\frac{a}{\sigma^2}\right)}{\sigma^2 \Gamma(1+2\eta)} [\frac{\Gamma(2 \eta)}{\Gamma(\eta-\frac{a}{\sigma^2}+1)} 
+ \frac{\Gamma(2 \eta)}{(\eta+\frac{a}{\sigma^2}) \Gamma(\eta-\frac{a}{\sigma^2})}]$$ which is equal to $\frac{1}{\sigma^2 \eta} [\frac{1}{\eta-\frac{a}{\sigma^2}}+\frac{1}{\eta+\frac{a}{\sigma^2}}] = \frac{2}{s^2}.$

\subparagraph{When $z_0 \to \infty$:}
Then $\Phi(aa,bb;z_0) \sim \frac{\Gamma(bb)}{\Gamma(aa)} e^{z_0} z_0^{a-bb} $ and $\Psi(aa, bb;z) \sim z_0^{-aa}$ so that
$$ A(z_0) \sim \frac{ 2 \Gamma\left(\eta-\frac{a}{\sigma^2}\right)}{\sigma^2} \frac{1}{z_0} [\frac{1}{z_0 \Gamma(\eta-\frac{a}{\sigma^2})} 
+ \frac{z_0}{(\eta+\frac{a}{\sigma^2})  \Gamma(\eta-\frac{a}{\sigma^2}+1)}]$$ which tends to $\frac{ 2 \Gamma\left(\eta-\frac{a}{\sigma^2}\right)}{\sigma^2 (\eta+\frac{a}{\sigma^2})  \Gamma(\eta-\frac{a}{\sigma^2}+1)}= \frac{2}{s^2}.$

The last piece is the  following result.			
\begin{lemma}
Let $a,b$ such that $a>0$ and $b>1$. Then, for all $z$
$$\Phi(a,b;z) \Psi(a+1,b;z) (b-a-1) + \Phi(a+1,b;z) \Psi(a,b;z) = \frac{\Gamma(b)}{a \Gamma(a)} z^{1-b} e^z  $$
\end{lemma}

\begin{proof}
To alleviate the notations let $\Phi \equiv \Phi(a,b;z)$ and $\Psi \equiv \Psi(a,b;z)$ and let us drop the dependency in $z$. We know that the Wronskian 
$\Phi \Psi' -\Psi \Phi'$ is given by $-\frac{\Gamma(b)}{a \Gamma(a)} z^{-b} e^z $. By substituting the expressions of the derivatives we get
$$\frac{\Gamma(b)}{a \Gamma(a)} z^{1-b} e^z = z(\frac{\Psi}{b} \Phi(a+1,b+1)+\Phi \Psi(a+1,b+1)  ) $$
so we want to prove the identity
$$ z(\Psi \Phi(a+1,b+1)+b\Phi \Psi(a+1,b+1)  ) =b(\Phi \Psi(a+1,b)(b-a-1)+\Psi \Phi(a+1,b))  $$
which in turn amounts to $\Psi (z\Phi(a+1,b+1)-b\Phi(a+1,b))= b \Phi ((b-a-1)\Psi(a+1,b)- z \Psi(a+1,b+1))$. Now by (\cite{DLMF}, 13.3.4) we have:
$$z\Phi(a+1,b+1)-b\Phi(a+1,b) =-b \Phi $$
and by (\cite{DLMF}, 13.3.10)
$$(b-a-1)\Psi(a+1,b)- z \Psi(a+1,b+1)=-\Psi $$
and the result follows.
\end{proof}

%%%%%%%%%%%%%%%%%%%%%%%%%%%%%%
%%%%%%%%%%%%%%%%%%%%%%%%%%%%%%
%%
%% Pricing vanilla options
%%
%%%%%%%%%%%%%%%%%%%%%%%%%%%%%%%
%%%%%%%%%%%%%%%%%%%%%%%%%%%%%%%%
\subsection{Pricing Vanilla Options}
We have all the elements to perform vanilla option pricing for the model. We focus on the computation of call option denoted 
\begin{equation*}
c(t,f_0)=e^{-rt}\mathbb{E}\left[\left(f_t-k\right)_+ \right].
\end{equation*}

We now take the Mellin transform $\mathcal{M}$ with respect to the strike as in \cite{jea2009} (see also \cite{panini2004}): the Mellin transform of the Call payoff with respect to the strike is given by $\int_0^\infty (x-k)_+ k^{\omega-1} dk = x\int_0^x k^{\omega-1} dk  - \int_0^x k^{\omega} dx = \frac{x^{\omega+1}}{\omega(\omega+1)}$,
for $\omega >0$. Therefore, we have
\begin{eqnarray}
\mathcal{M}(c(t,f_0),\lambda-1)= e^{-rt}\int_0^{+\infty} k^{\lambda-2} \mathbb{E}\left[\left(f_t-k\right)_+ \right] dk=\frac{e^{-rt}}{\lambda(\lambda-1)}\mathbb{E}\left[f_t^{\lambda}\right].
\end{eqnarray}
for $\lambda>1$.
If we take the Laplace transform of the above equation we get

\begin{eqnarray}
\int_0^{+\infty}e^{-\frac{s^2}{2}t}\int_0^{+\infty}k^{\lambda-2}\mathbb{E}\left[\left(f_t-k\right)_+ \right] dkdt=\frac{1}{\lambda(\lambda-1)}\int_0^{+\infty}e^{-\frac{s^2}{2}t}\mathbb{E}\left[f_t^{\lambda}\right]dt=\frac{g(\lambda,s)}{\lambda(\lambda-1)}
\end{eqnarray}

where $g$ is given in the previous section.  

\subsubsection{Strategy for the inversion of the double transform}

Let now $L(k, s)$ stands for the Laplace transform in time of the Call price. The Call price is given by the inverse Laplace transform of $L(k,s)$. We know that numerical algorithms like the Talbot method require only few (typically 20) evaluations of the function $L$, so our strategy will be to compute $L$ at the points required by the Talbot method by inverting the Mellin
transform of $L$. By Fubini's theorem the Mellin transform of $L(.,s)$ is given by:

$$\int_0^{+\infty}k^{\lambda-2} L(k,s) dk = \frac{g(\lambda,s)}{\lambda(\lambda-1)}$$

We shall need the following lemma.

\begin{lemma}
Let $s>0$. Then for $\lambda \in ]1, \lambda_+[$, 
$$c \to \frac{g(\lambda+i c,s)}{(\lambda+i c)(\lambda+i c-1)}$$ belongs to $L^1(\mathbb{R})$.
\end{lemma}

\begin{proof}
It follows readily from the fact that $\|\mathbb{E}\left[f_t^{\lambda+i c}\right]\|\leq \|\mathbb{E}\left[f_t^{\lambda}\right]\|$
\end{proof}

This lemma grants the validity of the inverse Mellin transform formula for $\lambda \in ]1, \lambda_+[$:
\begin{equation} \label{invmellin}
L(k, s) = \int_{\lambda+i \mathbb{R}} \frac{g(\tau,s)}{\tau (\tau-1)} k^{-\tau+1} d\tau.
\end{equation}

\subsubsection{Implementation}

In practice we discretize the integral \eqref{invmellin} using a quadrature with fixed size $N$. At each point, we use the hypergeometric series to evaluate $g$. It is readily checked that the
convergence of the series can be extended to the vertical line $\lambda+i \mathbb{R}$. We repeat this quadrature approximation for each point of the Talbot inversion algorithm. The choice of $N=100$ yields therefore typically 2000 calls of the function $g$. 

It should be noted that these calls can be performed in parallel. Note also that we can re-use the same evaluations of $g$ for different strikes $k$, so that the overall time to compute 
a whole (discretized) smile will be of the same order of magnitude than a single price, since the expensive part of the computation will be the evaluations of $g$.

%Then we should have
%\begin{eqnarray}
%\mathbb{E}\left[\left(f_t-k\right)_+ \right] = \frac{1}{(2\pi i)^2}\int_{\gamma-i\infty}^{\gamma+i\infty} k^{-(\lambda-1)} \int_{\tilde\gamma -i\infty}^{\tilde\gamma +i\infty}e^{zt}  g(\lambda,\sqrt{2z}) dzd\lambda .
%\end{eqnarray}

%Studying directly, because we know the expression of $g$, the convergence of the double integral  (and therefore determine $\gamma$ and $\tilde \gamma$) is complicated.\\

%In fact, if we denote $h(s)$ the right hand side of \eqref{Ltflambda} then 

%\begin{equation}
%\mathbb{E}\left[f_t^{\lambda}\right] = \frac{f_0^\lambda}{2\pi i} \int_{\tilde\gamma -i\infty}^{\tilde\gamma +i\infty}e^{st}h(s)ds.
%f_0^\lambda e^{-\frac{a}{\sigma^2}v_0 + \left(\frac{b}{\sigma^2}-\frac{\lambda\rho}{\sigma}e^{v_0}\right)}e^{-\frac{a^2t}{2\sigma^2}}F(t,v_0).
%\end{equation}

%and the option price is given by the inverse Mellin transform 

%\begin{eqnarray}
%c(t,f_0) = \frac{e^{-rt}}{2\pi i}\int_{\gamma-i\infty}^{\gamma+i\infty} \mathcal{M}(\lambda) k^{-(\lambda-1)} d\lambda.
%\end{eqnarray}

%Now we consider
%\begin{eqnarray}
%\int_0^{+\infty}e^{-\frac{s^2}{2}t}\int_0^{+\infty}k^{\lambda-2}\mathbb{E}\left[\left(f_t-k\right)_+ \right] dkdt=\frac{1}{\lambda(\lambda-1)}\int_0^{+\infty}e^{-\frac{s^2}{2}t}\mathbb{E}\left[f_t^{\lambda}\right]dt=g(\lambda,s)
%\end{eqnarray}

%\clearpage
%\input{Conclusion}

\section{Related Works}

Our work contributes to the literature aiming at overcoming the issues faced when implementing the affine model. The model proposed here is also presented in \cite{HenryLabordere2009} page 281 where it is called the Geometric Brownian, see also \cite{HenryLabordere2007}. The techniques used in \cite{HenryLabordere2009} are different from those used here (certainly they can be connected). Also, it seems to us that the problem of martingale property of the stock is not analysed for that particular model. Lastly, we don't know whether the formulas developed in this book lead to a reasonable numerical implementation. To illustrate the problem at stake and underline the usefulness of the series representation for I1 and I2 we just need to mention the fact the use of equations \eqref{heatKernelMorsePotential} and  \eqref{thetaFunc} (this function being the Hartman-Waston density) often lead to tedious numerical problems, see for example \cite{barrieu2004}.\\

We were able to obtain an explicit solution for the case $\alpha=1$ but we also established that the martingale property in that case depends on the parameter values. Extending the results to a general $\alpha$ is certainly of interest. If we understand \cite{HenryLabordere2009} in this general case the model might not be solvable.\\

Another work to which we are related is \cite{itkin2013} who studied a stochastic volatility model using Lie group analysis. He obtains a closed-form solution for the transition probability for the volatility process involving confluent hypergeometric functions. The author mainly focuses on volatility derivatives and the techniques used to derive his results are different form ours. Note also that the class of models consider in \cite{itkin2013} does not contain the model proposed here.

\section{Conclusion}

We propose a new stochastic volatility model for which we develop the key elements to perform equity and volatility derivatives pricing. We found the conditions on the parameters ensuring  the martingale property of the stock. For a particular set of parameter (i.e. $\alpha=1$) we compute the Mellin transform of the stock which enables the pricing of vanilla options. The model has, by construction, a volatility which is positive and therefore solve a major drawback of the traditional square root process, used for example in the \cite{heston} model, which imposes a constraint on the parameters (i.e., the Feller condition) that is not satisfied in practice.

%%%%%%%%%%%%%%%%%%%%%%%%
%%%%%%%%%%%%%%%%%%%%%%%%
%%
%% Biblio
%%
%%%%%%%%%%%%%%%%%%%%%%%%
%%%%%%%%%%%%%%%%%%%%%%%% 
\clearpage
\footnotesize
\linespread{0}
\selectfont
\bibliographystyle{abbrvnat}
%\bibliography{./HypergeometricStochasticVolatility}

\end{document}